\documentclass[11pt,onecolumn]{article}
\usepackage[utf8]{inputenc}
\usepackage{amssymb}
\usepackage{amsthm}
\usepackage{amsmath}
\usepackage{array}
\usepackage{booktabs}
\usepackage{graphicx}
\usepackage{color,soul}
\usepackage{empheq}
\usepackage[many]{tcolorbox}
\usepackage{mathtools}
\usepackage{makecell}
\usepackage{subfig}
\usepackage{stmaryrd}
\usepackage{enumitem}
\usepackage[vlined,linesnumbered,ruled,resetcount]{algorithm2e}
\usepackage{hyperref}
\usepackage{dsfont}
\usepackage{tikz}
\usepackage{cleveref}
\usepackage{thm-restate}

\usepackage[left=1.0in,top=1.0in,right=1.0in,bottom=1.0in,footskip=0.3in]{geometry}

\theoremstyle{plain}
\newtheorem{theorem}{Theorem}[section]

\newtheorem{lemma}[theorem]{Lemma}
\newtheorem{observation}[theorem]{Observation}

\newtheorem{corollary}[theorem]{Corollary}

\newtheorem{claim}[theorem]{Claim}

\newtheorem{remark}[theorem]{Remark}

\let\oldnl\nl% Store \nl in \oldnl
\newcommand{\nonl}{\renewcommand{\nl}{\let\nl\oldnl}}% Remove line number for one line
\newcommand{\eqdef}{{\stackrel{\mbox{\tiny \tt ~def~}}{=}}}

\SetKwProg{Procedure}{Procedure}{}{}
\SetKwRepeat{Do}{do}{while}

\makeatletter
\newcommand{\customlabel}[2]{%
   \protected@write \@auxout {}{\string \newlabel {#1}{{#2}{\thepage}{#2}{#1}{}} }%
   \hypertarget{#1}{#2}
}
\makeatother

\long\def\ignore#1{}
\DeclareMathOperator*{\avg}{avg}

\def\calI{{\cal I}}

\def\calD{{\cal D}}

\title{Matroid Base Packings: Improved Dynamic Matroid Density and Combinatorics of Tree Packings}
\author{Pavel Arkhipov \hspace{30pt} Vladimir Kolmogorov \\ \normalsize Institute of Science and Technology Austria \\ {\normalsize\tt $\{$pavel.arkhipov,vnk$\}$@ist.ac.at}}
\date{}

\begin{document}

\maketitle

\begin{abstract}
Greedy minimum-weight spanning tree packings are an important tool in graph connectivity algorithms. We study the corresponding process of greedy base packing in matroids, following the work of de Vos and Grilnberger.
Using a modified version of matroid base packings, we give a fully dynamic $(1 \pm \varepsilon)$-approximation to the matroid density using $O((\rho_{\max}^2\varepsilon^{-2}+\rho_{\max}\varepsilon^{-4})\log^3m_{\max})$ worst-case rank queries per update, where $\rho_{\max}$ upper-bounds the density and $m_{\max}$ upper-bounds the ground set size.
Sampling yields a $(1 \pm \varepsilon)$-approximation with high probability against an oblivious adversary using $O(\varepsilon^{-6}\log^6m_{\max})$ worst-case rank queries per update.

For graphic matroids, we strengthen the lower bound on the convergence rate of relative edge loads to ideal loads, closing the gap between the lower and upper bounds up to a logarithmic factor.
We also show that a packing of $O(\lambda^5\log m)$ trees contains a tree crossing some minimum cut once, improving the bound $O(\lambda^7\log^3m)$ of Thorup.
In the appendix, we consider a specialization of the greedy base packings to bicircular matroids, which yields a dynamic approximation of the graph density. For this, we develop a dynamic data structure that maintains a minimum-weight maximal pseudoforest.
\end{abstract}
\newpage

\section{Introduction}
A tree packing is a collection of spanning trees in a connected graph; the trees need not be edge-disjoint. One way to build such a packing is to start with weight zero on every edge, repeatedly choose a minimum-weight spanning tree (MST), and then increase the weight of each edge in that tree by one. This process is called a \emph{greedy tree packing}~\cite{thorup}, or just a \emph{tree packing} when the process with which it was obtained is clear. For every edge, the fraction of the trees that use this edge converges to a fixed limit as the number of trees goes to infinity, even though such tree packings are not unique in general, since there may be multiple minimum spanning trees. Greedy tree packings are used in algorithms for minimum cuts~\cite{karger_linear_mincut,thorup,christiansen}, minimum $k$-way cuts~\cite{thorup_k_way_cuts}, network unreliability~\cite{cen_unreliability}, and dynamic approximation of fractional arboricity~\cite{christiansen}.

A matroid $M=(E,\calI)$ consists of a finite ground set $E$ and a family $\calI\subseteq2^E$ of \emph{independent sets}. The empty set is independent, every subset of an independent set is independent, and if for $A,B\in\calI$ we have $|A|<|B|$, then some $e\in B\setminus A$ satisfies $A\cup\{e\}\in\calI$. A \emph{base} is a maximal independent set. For a set $H\subseteq E$, its \emph{rank} $r(H)$ is the size of a largest independent subset of $H$. The \emph{restriction} $M|H$ is the matroid with ground set $H$ and independent sets $\{I\in\calI:I\subseteq H\}$.
For a graph, taking the edge set as a ground set and the collection of all acyclic subgraphs as independent sets forms the \emph{graphic matroid}. When the graph is connected, spanning trees are bases of this matroid.

Similar to a greedy tree packing, one can consider a greedy base packing in some matroid. In a greedy base packing, the current weight of an element is the number of earlier bases containing it, and each new base is a minimum-weight base. For an element $e\in E$ and an integer $k\geq1$, the \emph{relative load} $\ell_e^k$ is the fraction of the first $k$ bases containing $e$. The vectors of relative loads $\ell^k \in \mathbb{R}^E$ converge to a vector $\ell^*$, the \emph{ideal relative loads}, independently of how ties are broken. For a nonempty matroid without loops, define its \emph{density} as
\begin{equation}\label{eq:intro_xmin_rho}
\rho(M):=\max_{\varnothing\ne H\subseteq E}\frac{|H|}{r(H)}.
\end{equation}
The identity $\rho(M)=1/\min_e\ell_e^*$ connects density and greedy base packings~\cite{de2025dynamic}. For graphic matroids, $\rho(M)$ is the fractional arboricity of the graph.

The same quantity $\rho(M)$ is also called the \emph{fractional covering number}, because the minimum number of bases that can cover $E$ is $\lceil\rho(M)\rceil$~\cite{edmonds}. In their paper, de Vos and Grilnberger write $\beta$ for the covering number.
They also study the \emph{fractional packing number} of a matroid, which is defined as
\begin{equation}\label{eq:intro-packing-number}
\Phi(M):=\min_{\substack{A\subseteq E\\ r(E\setminus A)<r(E)}}\frac{|A|}{r(E)-r(E\setminus A)}.
\end{equation}
The maximum number of pairwise disjoint bases of $M$ equals $\lfloor\Phi(M)\rfloor$~\cite{edmonds_base_packing}.

We consider the same dynamic matroid model as in~\cite{de2025dynamic}. An update inserts or deletes one element of the ground set. Deleting an element $e$ replaces $M$ by $M|(E\setminus\{e\})$. If an element $e\notin E$ is inserted, an adversary chooses a family $\calI_e$ of sets containing $e$ such that $(E\cup\{e\},\calI\cup\calI_e)$ is a matroid. Thus, insertion does not change which sets of old elements are independent. The current matroid can be accessed through a rank oracle, which returns $r(H)$ for a queried set $H$.

In this model, de Vos and Grilnberger~\cite{de2025dynamic} gave dynamic $(1\pm\varepsilon)$-approximation algorithms for the fractional packing and covering numbers, with the rank-query bounds in~\cref{tab:prior-dynamic-matroid-bounds}. Here $m_{\max}\geq2$ bounds the ground-set size throughout the update sequence, and $\Phi_{\max}$ and $\rho_{\max}$ bound the packing number and density, respectively. Their randomized guarantee holds with high probability against an oblivious adversary, whose update sequence is fixed independently of the algorithm's random choices.
\begin{table}[htbp]
\centering
\caption{Rank queries per update for the dynamic matroid algorithms of de Vos and Grilnberger~\cite{de2025dynamic}.}
\label{tab:prior-dynamic-matroid-bounds}
\begin{tabular}{@{}lcc@{}}
\toprule
Quantity and update guarantee & Deterministic & Randomized \\
\midrule
Packing, worst-case & $O(\Phi_{\max}^2\varepsilon^{-4}\log^3m_{\max})$ & --- \\
Packing, amortized & $O(\Phi_{\max}\varepsilon^{-4}\log^3m_{\max})$ & --- \\
Density, worst-case & $O(\rho_{\max}^2\varepsilon^{-4}\log^3m_{\max})$ & $O(\varepsilon^{-8}\log^6m_{\max})$ \\
\bottomrule
\end{tabular}
\end{table}
They developed a dynamic minimum-weight-base algorithm and maintained dynamic base packings. We build on their work to improve the dynamic density query bounds using a modified process that we call \emph{base packings with filtering}.

\subsection{Our contribution}

\paragraph*{Dynamic matroid density.} The main new algorithmic ingredient is \emph{base packing with filtering}: we pack bases greedily, but elements with large relative loads are omitted from later minimum-weight base computations. This process is described in more detail in~\cref{sec_filtering_description}. Writing $E_k$ for the set of elements retained after $k$ steps, we show that $1/\min_{e\in E_k}\ell_e^k$ still approximates $\rho(M)$; see~\cref{th:filtered-greedy}.
The packing with filtering yields a deterministic fully dynamic $(1+\varepsilon)$-approximation to $\rho(M)$ using $O((\rho_{\max}^2\varepsilon^{-2}+\rho_{\max}\varepsilon^{-4})\log^3m_{\max})$ worst-case rank queries per update, where $\rho_{\max}$ bounds the density and $m_{\max}$ bounds the ground-set size; see~\cref{th_dynamic_density}.
Subsampling removes the dependence on $\rho_{\max}$: against an oblivious adversary, we maintain a $(1\pm\varepsilon)$-approximation using $O(\varepsilon^{-6}\log^6m_{\max})$ worst-case rank queries per update, with estimates correct with high probability; see~\cref{th_dynamic_density_sampling}. These are query bounds for general matroids, not bounds on the time needed to answer the rank queries.
Compared with de Vos and Grilnberger~\cite{de2025dynamic}, we improve the deterministic bound from $O(\rho_{\max}^2\varepsilon^{-4}\log^3m_{\max})$ to $O((\rho_{\max}^2\varepsilon^{-2}+\rho_{\max}\varepsilon^{-4})\log^3m_{\max})$, and the randomized bound from $O(\varepsilon^{-8}\log^6m_{\max})$ to $O(\varepsilon^{-6}\log^6m_{\max})$. Filtering limits how many bases contain a given element, which reduces the number of calls to the dynamic minimum-weight base data structures.

In the appendix, we specialize the algorithm to bicircular matroids, obtaining an algorithm for dynamic graph density approximation. For this, we develop a data structure for maintaining a minimum-weight maximal pseudoforest in $O(\log n)$ worst-case time per edge update;\footnote{The $O(\log n)$ update time holds assuming that the number of edges is polynomially bounded. Otherwise, a correct bound on the update time is $O( \max \{\log n, \log m\})$.} see~\cref{sec:bicircular-matroids}.

\paragraph*{Ideal loads.} The base polytope $P_B\subseteq\mathbb{R}^E$ of $M$ is the convex hull of the indicator vectors of its bases. We show that $\ell^* \in \mathbb{R}^E$ minimizes $\sum_{e \in E}\varphi(x_e)$ over $P_B$ for every convex function $\varphi:[0,1]\to\mathbb{R}$; see~\cref{th_char_x_projection}.

\paragraph*{A lower bound for the load convergence.} Given a graph with $n$ vertices, $m$ edges, and edge connectivity $\lambda$, Thorup~\cite{thorup} proved that for $\varepsilon < 2$, a greedy tree packing with $k \geq 6 \lambda \log m / \varepsilon^2$ trees satisfies $\|\ell^k - \ell^*\|_\infty \leq \varepsilon / \lambda$. In terms of the number of trees $k$, the bound reads $\|\ell^k - \ell^*\|_\infty \leq O(\sqrt{\log m/(k\lambda)})$.
We show that this is tight up to a factor of $\sqrt{\log m}$: in \cref{th:lower-bound}, we construct multigraphs with greedy tree packings for which
\begin{equation}
\|\ell^k-\ell^*\|_\infty=\Omega\left(\frac{1}{\sqrt{k\lambda}}\right)
\end{equation}
for $k=O(n^2)$. Thus, the dependence on $k$ and $\lambda$ cannot be improved. Experiments on the same graph family suggest an error of order $1/\sqrt{k}$ even when random tiebreaking is used for the construction of MSTs.

\paragraph*{A tree crossing a minimum cut once.} Thorup's dynamic minimum-cut algorithm uses a greedy tree packing with $O(\lambda^7\log^3m)$ trees to guarantee that one of the trees crosses some minimum cut exactly once~\cite{thorup}. We show that $O(\lambda^5\log m)$ trees are enough; see \cref{th_lambda_5_respect}. We also observe a lower bound: there exist greedy tree packings for which none of the first $\Theta(\log n/\log\lambda)$ trees cross any min-cut exactly once; see~\cref{rem:respect-lower-bound}.

\subsection{Background and related work}

\paragraph*{Dynamic packings and density.} 
De Vos and Grilnberger~\cite[Proposition 1]{de2025dynamic} maintain a minimum-weight base under element insertions and deletions using $O(\log m_{\max})$ worst-case rank queries per update. A weight change can be implemented by deletion and reinsertion, with the same asymptotic query bound. This algorithm also applies to dynamic bases in restrictions of the current matroid, since a rank query in a restriction is a rank query on the same set in the original matroid.

De Vos and Grilnberger~\cite{de2025dynamic} proved that, for a nonempty loopless matroid, the packing number and density are related to the ideal loads by
\begin{equation}
\Phi(M)=\frac{1}{\max_{e\in E}\ell_e^*}
\qquad\text{and}\qquad
\rho(M)=\frac{1}{\min_{e\in E}\ell_e^*}.
\end{equation}
These identities underlie their algorithms summarized in~\cref{tab:prior-dynamic-matroid-bounds}. For the amortized packing bound, they update the prefix of greedy bases needed for the current packing estimate and defer updates to later bases until these become relevant.

Thus, there is asymmetry in the update times for $\rho$ and $\Phi$. Packing uses the maximum relative load, while density uses the minimum relative load. It is known that if each element is present in only a few bases, then dynamically updating a base packing is easier. This can be shown with a \emph{recourse} argument, introduced in~\cite{christiansen} for tree packings and used in~\cite{de2025dynamic} for base packings. With this technique, it is possible to remove a factor of $\Phi$ from the update time for a greedy base packing, if we know that all element loads are upper-bounded by $\Theta(1/\Phi)$. However, it is not possible to remove a factor of $\rho$ right away, since there may be elements with much higher loads than $1/\rho$. This observation motivates our modification to the greedy base packing process. When we find the next base, we omit elements whose load is higher than some threshold, which is chosen such that this never affects the elements that have minimum $\ell^*_e$. This allows us to reduce the maximal element load and use the recourse trick for bounding the update time.

For a graph containing a cycle, the density of its bicircular matroid equals the density of the graph. Thus, specializing to bicircular matroids gives dynamic graph-density algorithms, with forests handled separately. Graph-specific algorithms already give faster bounds than our sampled specialization. For graphs with $n$ vertices, Sawlani and Wang~\cite{SawlaniWang:STOC20} give $O(\varepsilon^{-6}\log^4n)$ worst-case update time, while Chekuri et al.~\cite{chekuri_density} give $O(\varepsilon^{-6}\log^3n\log(2+\rho(G)))$. Chekuri and Quanrud~\cite{ChekuriQuanrud:22} give $O(\varepsilon^{-4}\log^2n)$ amortized update time and the worst-case bound
\begin{equation}
O\left(\varepsilon^{-4}\log n\log(2+\rho(G))+\varepsilon^{-6}\log^3n\left(\log\log n+\log(1/\varepsilon)\right)\right).
\end{equation}
These algorithms are deterministic. Under the assumption $m_{\max}=n^{O(1)}$, our sampled specialization takes $O(\varepsilon^{-6}\log^6n)$ worst-case time per update and is correct with high probability against an oblivious adversary. The bicircular-matroid application illustrates how our general rank-query bounds translate into update-time bounds, rather than improving these graph-specific results.

\paragraph*{Ideal loads.} 
Harb, Quanrud, and Chekuri observed a connection between greedy tree packings and the Frank--Wolfe algorithm for minimizing $\|x\|_2^2$ over the spanning-tree polytope~\cite{Harb:ESA23}. Cen et al.~\cite{panigrahi_arboricity} showed that ideal loads also maximize entropy over the spanning-tree polytope. We point out that this follows from \cref{th_char_x_projection} by taking $\varphi(t)=t\log t$, with $\varphi(0)=0$. For general matroids, de Vos and Grilnberger~\cite{de2025dynamic} established the identity $\rho(M)=1/\min_e\ell_e^*$. A special case of this identity for graphic matroids and fractional arboricity appears in~\cite{christiansen}. We give a similar convergence result for packings with filtering.

\paragraph*{Greedy tree packings and minimum cuts.} 
Thorup's dynamic minimum-cut algorithm maintains $O(\lambda^7\log^3m)$ greedily packed trees and uses the fact that one of them crosses some minimum cut exactly once~\cite{thorup}.
Christiansen and de Vos obtained a dynamic minimum-cut algorithm using only $O(\lambda^3\log m)$ trees by using a modified technique: either a tree crosses a minimum cut once, or a minimum cut can be found as a trivial cut in some contracted graph~\cite{christiansen}. Therefore, although we improve the upper bound on the number of trees needed to 1-respect some min-cut to $O(\lambda^5\log m)$, our result does not improve the update time of~\cite{christiansen}. Improvement in the runtime for the dynamic min-cut would be possible if one could reduce the required number of trees to $o(\lambda^3 \log m)$.

The study of the convergence of $\ell^k$ to $\ell^*$ is also motivated by the dynamic min-cut applications. Both~\cite{thorup} and~\cite{christiansen} use upper bounds on the value $\|\ell^k - \ell^*\|_\infty$ in the analysis of their algorithms.
The previous lower bound on the load error was $\Omega(1/(k^{2/3}\lambda^{1/3}))$ for $k=O(\lambda\sqrt n)$ due to de Vos and Christiansen~\cite{christiansen}. Their construction uses adversarial tiebreaking, like ours. In the same paper, the authors show that there \emph{exists} a collection of $k$ spanning trees whose relative loads differ from the ideal loads by $O(1/k)$ in the infinity norm. However, it is unclear whether such packings can be constructed with a simple greedy procedure.

\section{Matroid density}\label{sec:general-matroids}
We study density of a matroid $M = (E, \calI)$. We assume that $M$ does not have loops, i.e. every set containing only one element is independent. Presence of loops would make the density infinite.

Recall the definition of greedy base packings. Let $B_1,B_2,\ldots$ be a sequence of bases, and let $\chi_B$ denote the indicator vector of a base $B$. Define the vectors of element loads and relative element loads by
\begin{equation}
L^k:=\chi_{B_1}+\cdots+\chi_{B_k}\in\mathbb{R}^{E}
\quad\text{and}\quad
\ell^k:=\frac{L^k}{k}\in\mathbb{R}^{E}.
\end{equation}
Then, $B_1,B_2,\ldots$ is a greedy base packing if $B_k$ is a minimum-weight base in $M$ with respect to weights $\ell^{k-1}$, where we take $\ell^0=\mathbf{0}$.

\subsection{Base packings with filtering and density estimation}\label{sec_filtering_description}

We now define \emph{greedy base packing with filtering}, used in \cref{th:filtered-greedy}. For this process, $m_{\max}\geq\max\{2,|E|\}$ may be any fixed upper bound on the ground-set size; when comparing packings before and after an element update, we use the same bound and the same interval $[\rho^-,\rho^+]$.

Given a matroid $M=(E,\calI)$ and input interval $[\rho^-,\rho^+]$, the algorithm
does the following for $k=0,1,2,\ldots$:
\begin{itemize}
\item Let $\ell^k=\frac{1}{k}(\chi_{B_1}+\cdots+\chi_{B_k})$ (or $\ell^0=\mathbf{0}$ if $k=0$)
\item 
% Let $E_{k}=E-\{e\in E:\ell^{k}_e>2/\rho^-\text{ and }k\ge 6\rho^+\log m_{\max}\}$.
Let $E_{k}=E-\{e\in E:\ell^{k}_e>2/\rho^-\}$ if $k\ge 6\rho^+\log m_{\max}$, otherwise $E_k = E$.
\item Let $B_{k+1}\subseteq E_k$ be a minimum-weight base in $M|E_k$ with respect to weights $\ell^k$.
\end{itemize}
The sets $E_k$ are not necessarily nested. If $\rho^-=0$ or $\rho^+ = \infty$, then we always have $E_k = E$, which yields the usual greedy base packing.

\begin{restatable}{theorem}{filteredGreedyConvergence}\label{th:filtered-greedy}
Let $M=(E,\calI)$ be a nonempty loopless matroid with density $\rho\in[\rho^-,\rho^+]$, and let $m_{\max}\geq\max\{2,|E|\}$. For every $k\geq1$, greedy base packing with filtering satisfies $\rho(M|E_1\cap\ldots\cap E_k)=\rho$, where $E_k$ is the retained set at step $k$. If $\varepsilon\in(0,1]$ and $k\geq20\rho\log(3m_{\max})/\varepsilon^2$, then
\begin{equation}
0\leq\frac{1}{\min_{e\in E_k}\ell_e^k}-\rho\leq\varepsilon\rho.
\end{equation}
\end{restatable}

We prove \cref{th:filtered-greedy} in \cref{sec:th:filtered-greedy}.
% using an argument similar to Thorup's proof of~\cite[Claim 16.1]{thorup} and to multiplicative-weights analysis.

\subsection{Dynamic matroid density}\label{sec:dynamic-matroid-density}
We consider the dynamic matroid model of de Vos and Grilnberger~\cite{de2025dynamic}. Deleting an element $e$ replaces $M$ by $M|(E - \{e\})$, and inserting an element gives a matroid extension such that the restriction to the old ground set is $M$. Thus, an update does not change the ranks of sets that do not contain the element being changed. We assume that the matroid is loopless, and have a convention that $\rho(M)=0$ when $E=\varnothing$. 

The complexity of the algorithms is measured in terms of the number of queries to the rank oracle of the current matroid. We assume that a rank query to a matroid takes at least $\Omega(1)$ time. This convention makes it more convenient to state the complexity of the algorithm, since we also need to spend time on ordinary operations, like maintaining a heap. The time spent on such operations is dominated by the rank queries, assuming each rank query cannot be done instantly and takes at least constant time.

\begin{theorem}\label{th_dynamic_density}
Let $\varepsilon\in(0,1/2]$, and let $M$ be a fully dynamic loopless matroid accessed through a rank oracle. Suppose that $m_{\max}\geq2$ is a fixed upper bound on its ground set size, and fix a parameter $\rho_{\max}\geq1$. There exists a deterministic data structure that uses
\begin{equation}
O\left((\rho_{\max}^2\varepsilon^{-2}+\rho_{\max}\varepsilon^{-4})\log^2 m_{\max}\right)
\end{equation}
calls to a min-weight dynamic matroid base data structure, or
\begin{equation}
O\left((\rho_{\max}^2\varepsilon^{-2}+\rho_{\max}\varepsilon^{-4})\log^3m_{\max}\right)
\end{equation}
rank queries per update in the worst case and returns a value $\hat{\rho}$ satisfying $\rho(M) \leq \hat{\rho} \leq (1+\varepsilon)\rho(M)$ whenever the current density satisfies $\rho(M)\leq\rho_{\max}$. The update complexity bound holds even when the density exceeds $\rho_{\max}$.
\end{theorem}

\paragraph*{Maintaining a base packing with filtering.}
Our main ingredient for the dynamic matroid density algorithm is the base packing with filtering. Below we show how to efficiently maintain it dynamically.

By~\cite[Proposition 1]{de2025dynamic}, a minimum-weight base can be maintained with $O(\log m_{\max})$ worst-case rank queries per element insertion or deletion. A weight change can be implemented by deleting and reinserting the element. Therefore, we assume three types of updates: element insertion, element deletion, and change of the element weight, all taking $O(\log m_{\max})$ rank queries. In the same way, we can also maintain a minimum-weight base in any restriction of $M$, since the rank of a set in the restriction equals its rank in $M$. We break ties using one fixed ordering of the elements. 

\begin{theorem}
    \label{th_maintaining_packing_complexity}
    Given the parameters $[\rho^-,\rho^+]$ and $k$, one can maintain the corresponding base packing in a dynamic matroid. One matroid update needs $O\left(k+k\rho^+\log m_{\max}+\frac{k^2}{\rho^-}\right)$ calls to a minimum-weight base update algorithm.

    In particular, if $\rho^-=\Theta(\rho^+)$ and $k=\left\lceil\frac{20 \rho^+ \log (3m_{\max})}{\varepsilon^2}\right\rceil$, as in~\cref{th:filtered-greedy}, then the number of base update calls is $O\left(\left(\rho^++\frac 1 {\varepsilon^2}\right)\cdot\frac{\rho^+\log^2 m_{\max}}{\varepsilon^2}\right)$.
\end{theorem}
\begin{proof}

The base $B_{j+1}$ is maintained in $M|E_j$ with weights $L^j$, the unnormalized loads. After an element update, go through the bases $B_1, B_2, \ldots$. At each step $j$, update the elements whose loads $L^j$ have changed. Such elements may also need to be added or removed from the set $E_j$. Then update $B_{j+1}$: call an update of the dynamic minimum weight base for every change in $E_j$, and also for every change in element weights $L^j$. It remains to bound the number of these dynamic updates. We use the structural lemma below.

\begin{lemma}
    \label{lem_structural_update}
    Let $M = (E, \calI)$ and $\widetilde{M} = (\widetilde{E}, \widetilde{\calI})$ be two matroids that differ by a single element $e_\circ$, i.e. $\widetilde{E} = E \sqcup \{e_\circ\}$ and $M = \widetilde{M} | E$. Let $\{B_j\}_{j=1}^k$ and $\{\widetilde{B}_j\}_{j=1}^k$ be greedy base packings with filtering in $M$ and $\widetilde{M}$ respectively, where we use the same parameters $[\rho^-,\rho^+]$ and $m_{\max}$ and the same fixed ordering on $E$ for tiebreaking.
    For each $j \in [k]$, let $L^j = \chi_{B_1}+\cdots+\chi_{B_j} \in\mathbb{Z}^{E}$ and $\widetilde L^j = \chi_{\widetilde B_1}+\cdots+\chi_{\widetilde B_j}\in\mathbb Z^{\widetilde E}$ be corresponding vectors of unnormalized loads. Then 
    \begin{equation}
        \label{eq_structural_update_theorem}
        \sum_{e \in E} |L^j_e - \widetilde{L}^j_e| \leq \widetilde{L}^j_{e_\circ}.
    \end{equation}
\end{lemma}
\begin{proof}
    Observe that $\sum_{e \in E} |L^j_e - \widetilde{L}^j_e| = \| L^j - \widetilde{L}^j |_E \|_1$. Set $L^0=\mathbf{0}$ and $\widetilde L^0=\mathbf{0}$. It is enough to prove, for every $0\leq j<k$, the inequality
    \begin{equation}
        \label{eq_structural_update_goal}
        \| L^{j+1} - \widetilde{L}^{j+1} |_E \|_1 \leq \| L^j - \widetilde{L}^j |_E \|_1 + \mathbf{1}_{\{e_\circ \in \widetilde{B}_{j+1}\}}.
    \end{equation}
    Summing these inequalities gives \eqref{eq_structural_update_theorem}.

    Fix $0\leq j<k$ and consider the map $\Lambda : \mathbb{Z}^E \to \mathbb{Z}^E$ defined as follows. For a vector $L \in \mathbb{Z}^E$, let $\Lambda(L) = L + \chi_B$, where $B$ is the minimum-weight base of $M | ( E-\{e\in E:L_e>j\cdot 2/\rho^-\text{ and }j\ge 6\rho^+\log m_{\max}\})$ under element weights $L$. So, $L^{j+1} = \Lambda(L^j)$.
    
    The key claim is that $\Lambda$ is a nonexpansive mapping in the 1-norm metric of $\mathbb{Z}^E$, i.e. $\|\Lambda(L) - \Lambda(L')\|_1 \leq \|L - L'\|_1$. The triangle inequality implies that it is enough to prove $\|\Lambda(L) - \Lambda(L')\|_1 \leq \|L - L'\|_1$ for $\|L - L'\|_1 = 1$. Without loss of generality, let $L' = L + \chi_{\{f\}}$ for some $f \in E$. Let $B$ and $B'$ be the corresponding minimum-weight bases, using the fixed tiebreaking order. Only $f$ can leave the retained set when passing from $L$ to $L'$.
    If $f \notin B$, then $B = B'$, and $\|\Lambda(L) - \Lambda(L')\|_1 = \|L + \chi_B - L' - \chi_B\|_1 = 1$. 
    If $f\in B$, then either $B'=B$, or $B'=B-\{f\}$, or $B'=B-\{f\}+\{g\}$ for some $g\ne f$. The second case occurs when filtering removes $f$ and decreases the rank of the restriction. In these three cases, $\Lambda(L)-\Lambda(L')$ equals $-\chi_{\{f\}}$, $\mathbf{0}$, or $-\chi_{\{g\}}$, respectively, so its 1-norm is at most one. Thus $\Lambda$ is a nonexpansive mapping.

    Using the triangle inequality, we have
    \begin{equation}
        \| L^{j+1} - \widetilde{L}^{j+1} |_E \|_1 \leq
        \| L^{j+1} - \Lambda(\widetilde{L}^{j} |_E) \|_1 +
        \| \Lambda(\widetilde{L}^{j} |_E) - \widetilde{L}^{j+1} |_E \|_1.
    \end{equation}
    For the first term on the right-hand side, we have $\| L^{j+1} - \Lambda(\widetilde{L}^{j} |_E) \|_1 = \| \Lambda(L^{j}) - \Lambda(\widetilde{L}^{j} |_E) \|_1 \leq \| L^{j} - \widetilde{L}^{j} |_E \|_1$ using the nonexpansion of $\Lambda$. 
    For the second term, the retained sets defined by $\widetilde L^j|_E$ in $M$ and by $\widetilde L^j$ in $\widetilde M$ agree on $E$. If $e_\circ\notin\widetilde B_{j+1}$, the fixed tiebreaking order therefore gives $\Lambda(\widetilde L^j|_E)=\widetilde L^{j+1}|_E$. Otherwise, deleting $e_\circ$ from the retained restriction either reduces its rank, leaving $\widetilde B_{j+1}-\{e_\circ\}$ as the selected base, or replaces $e_\circ$ by one element $f\in E$. Accordingly, $\Lambda(\widetilde L^j|_E)$ equals either $\widetilde L^{j+1}|_E$ or $\widetilde L^{j+1}|_E+\chi_{\{f\}}$. Hence $\|\Lambda(\widetilde L^j|_E)-\widetilde L^{j+1}|_E\|_1\leq\mathbf{1}_{\{e_\circ\in\widetilde B_{j+1}\}}$. Plugging in these bounds yields \eqref{eq_structural_update_goal} and proves the lemma.
\end{proof}

Now we can bound the number of dynamic minimum weight base update calls used for the update of $B_{j+1}$. For insertion, we use \cref{lem_structural_update} with $M$ being the current matroid and $e_\circ$ being the added element, and for deletion, we use \cref{lem_structural_update} with $\widetilde{M}$ being the current matroid and $e_\circ$ being deleted. 
In both cases, the number of elements whose load changed is at most $1 + \sum_{e \in E} |L^j_e - \widetilde{L}^j_e|$. For each such change, we call a minimum-weight base update, resulting in at most $1 + \sum_{e \in E} |L^j_e - \widetilde{L}^j_e|$ calls. By \cref{lem_structural_update}, we have at most $1 + \widetilde{L}^j_{e_\circ}$ calls.

Let $\widetilde E_i$ denote the retained set after $i$ steps in $\widetilde M$. We claim that
\begin{equation}
\widetilde{L}^j_{e_\circ} \leq 1+\max\{6\rho^+\log m_{\max}, 2j/\rho^-\}.
\end{equation}
Suppose that $\widetilde{L}^j_{e_\circ} > 0$ and let $i\leq j$ be the last index with $e_\circ\in\widetilde B_i$. If $i-1<6\rho^+\log m_{\max}$, then $\widetilde{L}_{e_\circ}^j\leq i<1+6\rho^+\log m_{\max}$. Otherwise, $e_\circ\in\widetilde E_{i-1}$ implies $\widetilde{L}_{e_\circ}^{i-1}\leq2(i-1)/\rho^-$, so $\widetilde{L}_{e_\circ}^j=1+\widetilde{L}_{e_\circ}^{i-1}\leq1+2j/\rho^-$.

Consequently, the total number of update calls is
\begin{equation}\label{eq:dynamic-base-updates}
O\left(k+\sum_{j=0}^{k-1} \widetilde{L}^j_{e_\circ} \right)
=O(k(1+\widetilde{L}^k_{e_\circ}))
=O\left(k+k\rho^+\log m_{\max}+\frac{k^2}{\rho^-}\right),
\end{equation}
which proves the first claim of the theorem. The second claim follows by plugging in the values for $k$, $\rho^+$, $\rho^-$.
\end{proof}

\paragraph*{Choosing the interval $[\rho^-, \rho^+]$.}
We now prove~\cref{th_dynamic_density}. For $\rho_i=2^{i-1}$ and $i=0,\ldots,\lceil\log_2\rho_{\max}\rceil+1$, maintain a packing with filtering with interval $[\rho_i,4\rho_i]$ and $k=\Theta( \rho_i\log(m_{\max})/\varepsilon^2)$ bases. If $\rho(M)$ belongs to this interval, then~\cref{th:filtered-greedy} shows that the inverse minimum relative load over $E_k$ is between $\rho(M)$ and $(1+\varepsilon)\rho(M)$.

To select a correct interval, we maintain an unfiltered greedy packing of $\Theta( \rho_{\max}\log(m_{\max}) )$ bases, which gives us a factor of $2$ estimation for the density. If $\rho(M)\leq\rho_{\max}$, the inverse minimum relative load $q$ satisfies $\rho(M)\leq q\leq2\rho(M)$. Choose $i$ such that $\rho_i\leq q/2<2\rho_i$. Then $[q/2,q]\subseteq[\rho_i,4\rho_i]$, so we have a valid estimate. If the promised bound is exceeded, the estimate may be arbitrary, but the update bound still holds. This convention will later be useful for sampling.

By \cref{th_maintaining_packing_complexity}, we need $O((\rho_i^2\varepsilon^{-2}+\rho_i\varepsilon^{-4})\log^2m_{\max})$ minimum-weight base updates per matroid update. Multiplying by the $O(\log m_{\max})$ rank queries per base update and summing over $i$ gives
\begin{equation}
O\left((\rho_{\max}^2\varepsilon^{-2}+\rho_{\max}\varepsilon^{-4})\log^3m_{\max}\right)
\end{equation}
rank queries. The packing without filtering requires $O(\rho_{\max}^2\log^3m_{\max})$ rank queries and does not increase this bound. 

Finally, a heap of the loads of elements in $E_k$ maintains the minimum needed for the output. It takes $O(\log m_{\max})$ time per update in the loads, which is dominated by the rank queries, assuming each rank query takes at least constant time. This proves~\cref{th_dynamic_density}.

\subsection{Removing the density dependence by sampling}\label{sec:matroid-density-sampling-proof}
We apply the sampling technique to our deterministic algorithm, closely following~\cite{christiansen,de2025dynamic}. We include the parameter calculation because the cost of each deterministic instance is different.

\begin{theorem}\label{th_dynamic_density_sampling}
For $\varepsilon\in(0,1/2]$, there exists a randomized data structure that maintains an estimate between $(1-\varepsilon)\rho(M)$ and $(1+\varepsilon)\rho(M)$ for a fully dynamic loopless matroid, using $O(\varepsilon^{-6}\log^6m_{\max})$ worst-case rank queries per update, where $m_{\max} \geq 2$ is an upper bound on the number of elements in the matroid at any time. The estimates are correct with high probability throughout any polynomial-length update sequence against an oblivious adversary.
\end{theorem}
\begin{proof}
The sampling analysis of de Vos and Grilnberger~\cite{de2025dynamic} gives the following bound. Sample a set $H \subseteq E$ by choosing every element independently with probability $p$. If $p\rho(M)\geq C\varepsilon^{-2}\log m_{\max}/2$, where $C$ is a sufficiently large constant, then
\begin{equation}\label{eq:matroid-density-sampling}
(1-\varepsilon/3)p\rho(M)\leq\rho(M|H)\leq(1+\varepsilon/3)p\rho(M)
\end{equation}
with failure probability $m_{\max}^{-\Omega(1)}$. The power $\Omega(1)$ can be made arbitrarily large by adjusting $C$.

For each $i=0,\ldots,\lceil\log_2(2m_{\max})\rceil$, maintain a subset $H_i \subseteq E$ sampled with probability $p_i=\min\{1,C\varepsilon^{-2}\log m_{\max}/2^i\}$. For $i = 0$, we have $p_0 = 1$ and $H_0 = E$. For each $i$, run the algorithm from~\cref{th_dynamic_density} with accuracy $\varepsilon/3$ and density bound $10C\varepsilon^{-2}\log m_{\max}$. The idea is that if the true density is of order $2^i$, and level $i$ estimates the density of $M | H_i$ as $\hat{\rho}_i$, then returning $\hat{\rho}_i / p_i$ is a correct estimation of $\rho$.
When a new element is inserted, sample its memberships in $H_i$ independently; keep them until the deletion of the element. Then at any fixed time, each $H_i$ is an independent element sample of the current matroid. A union bound over all $i$ and any fixed polynomial-length sequence of updates makes~\eqref{eq:matroid-density-sampling} hold simultaneously whenever $p_i\rho(M)\geq C\varepsilon^{-2}\log m_{\max}/2$, with high probability. Initialize at the empty matroid with the density estimate zero.

To choose a level, use the estimate from the previous update, as in~\cite{de2025dynamic}. A single element insertion or deletion changes the density of a loopless matroid by at most one. Let $\hat{\rho}_{\mathrm{old}}$ be our previous estimate of the previous density $\rho_{\mathrm{old}}$ that satisfied $(1-\varepsilon)\rho_{\mathrm{old}}\leq \hat{\rho}_{\mathrm{old}} \leq(1+\varepsilon)\rho_{\mathrm{old}}$.
If $\hat{\rho}_{\mathrm{old}} \leq4C\varepsilon^{-2}\log m_{\max}$, the new density is at most $4C\varepsilon^{-2}\log m_{\max}/(1-\varepsilon)+1<10C\varepsilon^{-2}\log m_{\max}$, and we use the original matroid corresponding to the instance at $i=0$. Its output satisfies the desired accuracy bound. Otherwise choose $i$ with $2^i\leq \hat{\rho}_{\mathrm{old}} <2^{i+1}$. Since $\varepsilon\leq1/2$ and the density changes by at most one, the new density satisfies
\begin{equation}
\frac{C\log m_{\max}}{2\varepsilon^2}\leq p_i\rho(M)\leq\frac{5C\log m_{\max}}{\varepsilon^2}.
\end{equation}
Thus~\eqref{eq:matroid-density-sampling} applies. Moreover, the density of $M|H_i$ is less than $10C\varepsilon^{-2}\log m_{\max}$, so the deterministic algorithm instance at level $i$ is correct. Return $\hat{\rho}_i/p_i$. \Cref{th_dynamic_density} and~\eqref{eq:matroid-density-sampling} imply
\begin{equation}
(1-\varepsilon)\rho(M)\leq(1-\varepsilon/3)\rho(M)\leq\frac{\hat{\rho}_i}{p_i}\leq(1+\varepsilon/3)^2\rho(M)\leq(1+\varepsilon)\rho(M),
\end{equation}
which completes the step.

Substituting the density bound $10C\varepsilon^{-2}\log m_{\max}$ and accuracy $\varepsilon/3$ into~\cref{th_dynamic_density}, each level uses $O(\varepsilon^{-6}\log^5m_{\max})$ rank queries per update (even if the corresponding upper bound on the density does not hold). Updating all $O(\log m_{\max})$ levels gives the claimed bound.
\end{proof}

\section{Structural results}\label{sec:graphic-matroids}

In this section, we give a short characterization of the vector of ideal loads for greedy base packings. Then we specialize to graphic matroids and study tree packings. We give a lower bound on the load convergence and improve the number of trees sufficient to find a tree that crosses some minimum cut exactly once.

\subsection{Greedy base packings and ideal loads}\label{sec:greedy-base-packings}

Let $M=(E,\calI)$ be a matroid with rank function $r(\cdot)$. The base polytope $P_B\subset\mathbb R^E$ is the convex hull of the indicator vectors of its bases.
Another characterization of $P_B$ is the following:
\begin{equation}
    P_B = \{x \in \mathbb{R}^E_{\geq 0} \; | \; x(E) = r(E) \text{ and } x(H) \leq r(H) \quad \forall H \subseteq E \}.
\end{equation}

\begin{restatable}{theorem}{convexLimitCharacterization}\label{th_char_x_projection}
The limit $\lim_{k\to\infty}\ell^k = \ell^*$ exists and minimizes $\sum_e\varphi(x_e)$ over $P_B$ for every continuous convex function $\varphi:[0,1]\to\mathbb{R}$.
\end{restatable}
\begin{proof}
We start with a short recap of the Frank--Wolfe algorithm, which dates back to \cite{frank_wolfe_original}.
It is an iterative process for minimizing a strictly convex differentiable
function $f : P \to \mathbb{R}$ over a convex compact set $P \subseteq \mathbb{R}^d$. Start from any point $x^0 \in P$ and keep doing the following:
    \begin{enumerate}
        \item Choose $v^k \in \arg \min\limits_{v \in P} \left< \nabla f(x^k), v \right>$
        \item $x^{k+1} = (1 - \eta_k) x^k + \eta_k v^k$
    \end{enumerate}
Here, $\eta_k$ is the step size parameter. With a suitable choice of $\eta_k$, we get the convergence $x^k \to x^*$, where $x^* \in P$ is the minimum of $f$. For an introduction and surveys, we redirect the reader to \cite{pokutta_fw_inroduction}, \cite{global_linear_convergence_fw}, \cite{bomze_fw_survey}.

Now we proceed with the proof. First suppose that $\varphi$ is smooth and strictly convex, and let $f(x)=\sum_e\varphi(x_e)$, which is also strictly convex. We claim that the greedy base packing and the Frank--Wolfe algorithm with $\eta_k = \frac{1}{k+1}$ are identical in this setting. 

Let $v^k$ be the indicator vector of $B_{k+1}$. It is a minimum weight base for the weights $\ell^k$, therefore $\left< v^k, \ell^k \right> = \min_{v : \text{ vertex of } P_B} \left< v, \ell^k \right> = \min_{v \in P_B} \left< v, \ell^k \right>$.
We observe that the derivative $\varphi'$ is strictly increasing, so the coordinates of $\nabla f(\ell^k)$ have the same ordering, including ties, as the coordinates of $\ell^k$. Then $v^k$ minimizes not only $\left< v, \ell^k \right>$ but also $\left< v, \nabla f(\ell^k) \right>$, since the minimum-weight base depends only on the order of the elements.
For the sequence of $\ell^k$, we have the update rule $\ell^{k+1} = \frac{k}{k+1} \ell^k + \frac{1}{k+1} v^{k}$, which is precisely the Frank-Wolfe process for the function $f$ and the feasible set $P_B$ with the step size $\eta_k = \frac{1}{k+1}$. This step size implies convergence to the minimum of $f$ in $P_B$~\cite[Section 2.3]{freund_fw_averaging}, which implies the lemma.

For an arbitrary convex $\varphi$, choose smooth strictly convex functions converging uniformly to $\varphi$ on $[0,1]$. Each approximation is minimized at $\ell^*$, and $P_B\subseteq[0,1]^E$, so taking the limit gives $\sum_e\varphi(\ell_e^*)\leq\sum_e\varphi(x_e)$ for every $x\in P_B$.
\end{proof}

\paragraph*{A hierarchy of ideal loads.} Taking $\varphi(t)=t^2$ in \cref{th_char_x_projection}, we see that $\ell^*$ is the minimum-norm point of $P_B$. Fujishige's characterization of this point~\cite[Theorem 9.3]{Fujishige:book} gives the following description for a nonempty loopless matroid. In the current matroid, choose the inclusion maximal nonempty set $H$ that maximizes $|H|/r(H)$, assign $\ell_e^*=r(H)/|H|$ to each $e\in H$, and contract $H$. Repeat until every element has been assigned a load, taking $r$ at each step to be the rank function of the current matroid. Contracting $H$ gives a matroid on the remaining elements whose rank on $A$ is $r(A\cup H)-r(H)$. The assigned loads strictly increase from step to step. In particular, the first step assigns $1/\rho(M)$, giving another way to obtain $\rho(M)=1/\min_e\ell_e^*$, as also proved by de Vos and Grilnberger~\cite{de2025dynamic}.

\subsection{\texorpdfstring{A tight lower bound for $\|\ell^k - \ell^*\|_\infty$ in graphic matroids}{A tight lower bound for the infinity-norm load error}}

\begin{theorem}\label{th:lower-bound}
For all integers $n\geq3$ and $w\geq1$, there is an $n$-vertex multigraph of edge connectivity $\lambda=3w$ and a greedy tree packing such that, for every integer $1\leq k<n^2/4$,
\begin{equation}
    \|\ell^k-\ell^*\|_\infty=\Omega\left(\frac{1}{\sqrt{k\lambda}}\right),
\end{equation}
where $\Omega(\cdot)$ hides a universal constant.
\end{theorem}

We first analyze a packing on a graph of edge connectivity $3$. Then we duplicate the edges to obtain examples with higher $\lambda$.

\paragraph*{Construction and ideal loads.} For an integer $n\ge 3$, let $G_n$ be a multigraph with the vertex set $\{1,\ldots,n\}$.
The edge set consists of a path on the vertices $2,\ldots,n$, and edges connecting the vertex $1$ to all other vertices. Edges $(1,2)$ and $(1,n)$ are doubled. We denote the edges $(1,i)$ by $e_i$, and the edges $(i,i+1)$ by $f_i$, with second copies of $e_2$ and $e_n$ called $e'_2$ and $e'_n$ respectively. See \cref{fig:fan-lower-bound}.

\begin{figure}[!ht]
    \centering
    \begin{tikzpicture}[xscale=1.08,yscale=0.88,
        every node/.style={font=\small}]
        \node[circle,draw,fill=cyan!15,inner sep=1.6pt] (v1) at (3,2.5) {$1$};
        \node[circle,draw,fill=cyan!15,inner sep=1.6pt] (v2) at (0,0) {$2$};
        \node[circle,draw,fill=cyan!15,inner sep=1.6pt] (v3) at (1.2,0) {$3$};
        \node[circle,draw,fill=cyan!15,inner sep=1.6pt] (v4) at (2.4,0) {$4$};
        \node[circle,draw,fill=cyan!15,inner sep=1.6pt] (v5) at (3.6,0) {$5$};
        \node[circle,draw,fill=cyan!15,inner sep=1.6pt] (v6) at (4.8,0) {$6$};
        \node[circle,draw,fill=cyan!15,inner sep=1.6pt] (v7) at (6,0) {$7$};

        \draw (v2) -- node[below] {\scriptsize $f_2$} (v3);
        \draw (v3) -- node[below] {\scriptsize $f_3$} (v4);
        \draw (v4) -- node[below] {\scriptsize $f_4$} (v5);
        \draw (v5) -- node[below] {\scriptsize $f_5$} (v6);
        \draw (v6) -- node[below] {\scriptsize $f_6$} (v7);

        \draw (v1) -- node[midway,fill=white,inner sep=1.2pt]
            {\scriptsize $e_3$} (v3);
        \draw (v1) -- node[midway,fill=white,inner sep=1.2pt]
            {\scriptsize $e_4$} (v4);
        \draw (v1) -- node[midway,fill=white,inner sep=1.2pt]
            {\scriptsize $e_5$} (v5);
        \draw (v1) -- node[midway,fill=white,inner sep=1.2pt]
            {\scriptsize $e_6$} (v6);
        \draw (v1) to[bend right=22]
            node[midway,fill=white,inner sep=1.2pt]
            {\scriptsize $e_2$} (v2);
        \draw (v1) -- node[midway,fill=white,inner sep=1.2pt]
            {\scriptsize $e'_2$} (v2);
        \draw (v1) -- node[midway,fill=white,inner sep=1.2pt]
            {\scriptsize $e_7$} (v7);
        \draw (v1) to[bend left=22]
            node[midway,fill=white,inner sep=1.2pt]
            {\scriptsize $e'_7$} (v7);
    \end{tikzpicture}
    \caption{The multigraph $G_7$. The edges are ordered as $e_2,e'_2,f_2,e_3,f_3,e_4,\ldots,f_{n-1},e_n,e'_n$.}
    \label{fig:fan-lower-bound}
\end{figure}

The edge connectivity of $G_n$ is $3$; the number of edges is $m = 2n-1$. We next show that every edge has the same ideal load. Observe that $\frac{|E(G_n)|}{|V(G_n)|-1} = 2+\frac{1}{n-1}$. We will show that the whole graph is the witness for fractional arboricity. Consider any proper vertex set $S \subsetneq V$ containing $t \geq 2$ vertices. If vertex $1$ is in the set, then the induced subgraph has at most $2(t-1)$ edges. If vertex $1$ is not included, then $S$ induces a subgraph of the path and has at most $t-1$ edges. In both cases, $\frac{|E[S]|}{|S|-1} < \frac{|E(G_n)|}{|V(G_n)|-1}$. Since $E(G_n)$ attains the density, the density identity gives $\min_g\ell^*_g=(n-1)/(2n-1)$. The ideal loads sum to $r(E(G_n))=n-1$, so all $2n-1$ loads must equal this minimum. Thus, for every edge $g\in E(G_n)$,
\begin{equation}\label{eq:fan-ideal-load}
    \ell^*_g=\frac{n-1}{2n-1} = \frac{1}{2} - \frac{1}{4n-2}.
\end{equation}

\paragraph*{The packing process.} We use the following order on the edges:
\begin{equation}\label{eq:fan-edge-order}
    e_2,\ e'_2,\ f_2,\ e_3,\ f_3,\ e_4,\ \ldots,\
    f_{n-1},\ e_n,\ e'_n.
\end{equation}
At every iteration, the next minimum spanning tree is computed with Kruskal's algorithm. It goes through the edges ordered increasingly by the load, with tiebreaks decided by the order in \eqref{eq:fan-edge-order}. We refer to an edge's place in \eqref{eq:fan-edge-order} as its position, numbered from right to left. Thus, $e'_n$ has position $1$ and $e_2$ has position $2n-1$.

We group the trees into consecutive pairs: $(T_1,T_2),(T_3,T_4),\ldots$. In each pair, the two trees will turn out to be edge-disjoint. Since $m=2n-1$, the two trees cover all but one edge of the graph. These unused edges accumulate at the right end of the edge order, and the rightmost edges become underpacked.

We partition the pairs of trees into phases, where phase $r$ contains $4r$ pairs. For example, the first phase contains the trees $T_1, \ldots, T_8$. \Cref{fig:fan-staircase} illustrates the states of edge loads at the ends of some phases.

\begin{figure}[!ht]
    \centering
    \begin{tikzpicture}[
        x=1.35cm,y=0.62cm,
        phase edge/.style={draw=black!85},
        phase vertex/.style={circle,draw,fill=cyan!15,inner sep=1.4pt},
        phase load/.style={fill=white,inner sep=1.2pt,font=\scriptsize}
    ]
        \newcommand{\fanphase}[8]{%
            \begin{scope}[shift={(0,#1)}]
                \node[anchor=east,font=\small] at (-0.55,1)
                    {phase $#2$: $k=#3$};
                \coordinate (p#2-1) at (3.5,2);
                \foreach \i in {2,...,9} {
                    \pgfmathsetmacro{\xx}{\i-2}
                    \coordinate (p#2-\i) at (\xx,0);
                }

                % The edges are drawn in the order from
                % \eqref{eq:fan-edge-order}.
                \draw[phase edge] (p#2-1) to[bend right=20]
                    node[midway,phase load] {#4} (p#2-2);
                \draw[phase edge] (p#2-1) --
                    node[midway,phase load] {#5} (p#2-2);
                \draw[phase edge] (p#2-2) --
                    node[midway,phase load] {#5} (p#2-3);
                \draw[phase edge] (p#2-1) --
                    node[midway,phase load] {#5} (p#2-3);
                \draw[phase edge] (p#2-3) --
                    node[midway,phase load] {#5} (p#2-4);

                \draw[phase edge] (p#2-1) --
                    node[midway,phase load] {#6} (p#2-4);
                \draw[phase edge] (p#2-4) --
                    node[midway,phase load] {#6} (p#2-5);
                \draw[phase edge] (p#2-1) --
                    node[midway,phase load] {#6} (p#2-5);
                \draw[phase edge] (p#2-5) --
                    node[midway,phase load] {#6} (p#2-6);

                \draw[phase edge] (p#2-1) --
                    node[midway,phase load] {#7} (p#2-6);
                \draw[phase edge] (p#2-6) --
                    node[midway,phase load] {#7} (p#2-7);
                \draw[phase edge] (p#2-1) --
                    node[midway,phase load] {#7} (p#2-7);
                \draw[phase edge] (p#2-7) --
                    node[midway,phase load] {#7} (p#2-8);

                \draw[phase edge] (p#2-1) --
                    node[midway,phase load] {#8} (p#2-8);
                \draw[phase edge] (p#2-8) --
                    node[midway,phase load] {#8} (p#2-9);
                \draw[phase edge] (p#2-1) --
                    node[midway,phase load] {#8} (p#2-9);
                \draw[phase edge] (p#2-1) to[bend left=20]
                    node[midway,phase load] {#8} (p#2-9);

                \foreach \i in {1,...,9}
                    \node[phase vertex] at (p#2-\i) {};
            \end{scope}%
        }

        % These values are produced by simulate_fan_loads.py.
        \fanphase{9.6}{1}{8} {+0}{+0}{+0}{+0}{-1}
        \fanphase{6.4}{2}{24}{+0}{+0}{+0}{-1}{-2}
        \fanphase{3.2}{3}{48}{+0}{+0}{-1}{-2}{-3}
        \fanphase{0}  {4}{80}{+0}{-1}{-2}{-3}{-4}
    \end{tikzpicture}
    \caption{Edge load states at the ends of the first four phases for $G_9$. Each edge $g$ is labeled by $L^k_g-k/2$; negative values indicate underpacking.}
    \label{fig:fan-staircase}
\end{figure}

For even $k$, we describe a string of loads in the order \eqref{eq:fan-edge-order} by a \emph{state} vector: $\langle c_0,\ldots,c_s\rangle$ denotes $c_h$ consecutive copies of $k/2-h$ for each $0\le h\le s$. In the brackets, $c^{\times q}$ denotes $q$ consecutive entries equal to $c$, and is omitted when $q=0$. For example, in \cref{fig:fan-staircase}, the state after phase $1$ is $\langle13,4\rangle$, and after phase $4$, it is $\langle1,4,4,4,4\rangle=\langle1,4^{\times4}\rangle$. The following lemma describes the states throughout each phase.

\begin{lemma}\label{lem:fan-phase-evolution}
For every phase $r$ satisfying $2r+1\leq n$, the states immediately before its $4r$ pairs of trees are as given in \cref{tab:fan-deficit-states}. Each pair consists of two edge-disjoint trees, and the state at the end of the phase is $\langle z-4,4^{\times r}\rangle$, where $z=2n-4r+3$.
\end{lemma}

The rows of the table occur in the displayed order, and the index $a$ runs from $0$ to $r-1$ four times. Rows with an empty range are skipped.

\begin{table}[!ht]
\caption{Load states during phase $r$, with $z = 2n-4r+3$. The index $a$ runs from $0$ to $r-1$ four times, as phase $r$ consists of $4r$ pairs.}
\centering
\renewcommand{\arraystretch}{1.2}
\begin{tabular}{ccl}
\toprule
Part & Range & State before the next pair of trees\\
\midrule
I & $a=0$ & $\langle z,4^{\times(r-1)}\rangle$\\
I & $1\le a<r$ &
$\langle z,4^{\times(r-a-1)},3,4^{\times(a-1)},1\rangle$\\
II & $a=0$ & $\langle z-1,4^{\times(r-1)},1\rangle$\\
II & $1\le a<r$ &
$\langle z-2,4^{\times(a-1)},5,4^{\times(r-a-1)},1\rangle$\\
III & $a=0$ & $\langle z-2,4^{\times(r-1)},2\rangle$\\
III & $1\le a<r$ &
$\langle z-2,4^{\times(r-a-1)},3,4^{\times(a-1)},3\rangle$\\
IV & $a=0$ & $\langle z-3,4^{\times(r-1)},3\rangle$\\
IV & $1\le a<r$ &
$\langle z-4,4^{\times(a-1)},5,4^{\times(r-a-1)},3\rangle$\\
\bottomrule
\end{tabular}

\label{tab:fan-deficit-states}
\end{table}

\begin{proof}[Proof of \cref{lem:fan-phase-evolution}]
We use induction with \cref{tab:fan-deficit-states} being the induction hypothesis. The base case is $k=0$ and $r=1$. We have $z=2n-1$, so the first state is $\langle2n-1\rangle$, the initial state with zero loads.

Assume that the current state is the one in the corresponding row of the table. Run Kruskal's algorithm, add the resulting tree, and then run Kruskal's algorithm again. For each of the eight row types in the table, a direct check shows that the two trees are edge-disjoint and that packing them changes the current state into the next applicable state. Below we justify transitions in Part I; the other transitions are checked in the same way.

Consider the Part I row with $1\leq a<r$. Define $c=n-2r+2$ and $b=n-2a$; the condition $2r+1\leq n$ ensures that $c\geq3$. Starting from the state $\langle z,4^{\times(r-a-1)},3,4^{\times(a-1)},1\rangle$, the two runs of Kruskal's algorithm produce the following trees:
\begin{equation}\label{eq:fan-part-I-trees}
\begin{aligned}
T=& \{e'_n,e_2\}\cup\{f_i:2\leq i\leq c-1\}\cup\bigcup_{t=1}^{r-a-1}\{e_{b-2t+1},f_{b-2t+1}\}
\cup\bigcup_{t=0}^{a-1}\{e_{b+2t+1}\}\cup\bigcup_{t=1}^{a-1}\{f_{b+2t}\},\\
T'=& \{e'_2,e_b\}\cup\{e_i:3\leq i\leq c-1\}\cup\bigcup_{t=1}^{r-a-1}\{e_{b-2t},f_{b-2t}\}
\cup\bigcup_{t=0}^{a-1}\{f_{b+2t+1},e_{b+2t+2}\},
\end{aligned}
\end{equation}
where a union over an empty range is omitted. These trees are edge-disjoint and together contain every edge except $f_b$. Adding them to the tree packing increases $k/2$ and the load of every used edge by one, so the positions of the used edges in the state do not change. The unused edge $f_b$ moves down by one load level. Consequently, the state changes as follows:
\begin{equation}
\langle z,4^{\times(r-a-1)},3,4^{\times(a-1)},1\rangle\longrightarrow
\begin{cases}
\langle z,4^{\times(r-a-2)},3,4^{\times a},1\rangle,&a<r-1,\\
\langle z-1,4^{\times(r-1)},1\rangle,&a=r-1.
\end{cases}
\end{equation}
These are respectively the next row in Part I and the first row in Part II in \cref{tab:fan-deficit-states}. 

\cref{fig:fan-worked-pair} illustrates a transition in Part II with $r=3$ and $a=1$. In this example, the transition is from $\langle2n-11,5,4,1\rangle$ to $\langle2n-11,4,5,1\rangle$, and the only unused edge is $e_{n-2}$.

\begin{figure}[!ht]
    \centering
    \begin{tikzpicture}[
        scale=0.98,transform shape,x=1.35cm,y=0.68cm,
        worked first/.style={draw=blue,line width=0.8pt},
        worked second/.style={draw=red,line width=0.8pt},
        worked unused/.style={draw=black,densely dashed,line width=0.8pt},
        worked edge/.style={draw=black},
        worked vertex/.style={circle,draw,fill=cyan!15,inner sep=1.4pt},
        worked load/.style={fill=white,inner sep=1.2pt,font=\scriptsize}
    ]
        \newcommand{\workedstate}[5]{%
            \begin{scope}[shift={(0,#2)}]
                \node[anchor=east,font=\small,align=right] at (-0.55,1)
                    {#3,\\$k=#4$};
                \coordinate (#1-1) at (4.25,2);
                \coordinate (#1-2) at (0,0);
                \coordinate (#1-3) at (1,0);
                \coordinate (#1-stub-left) at (1.35,0);
                \coordinate (#1-dots) at (1.75,0);
                \coordinate (#1-stub-right) at (2.15,0);
                \coordinate (#1-4) at (2.5,0);
                \coordinate (#1-5) at (3.5,0);
                \coordinate (#1-6) at (4.5,0);
                \coordinate (#1-7) at (5.5,0);
                \coordinate (#1-8) at (6.5,0);
                \coordinate (#1-9) at (7.5,0);
                \coordinate (#1-10) at (8.5,0);

                \draw[worked edge] (#1-1) to[bend right=20]
                    node[midway,worked load] {+0} (#1-2);
                \draw[worked edge] (#1-1) --
                    node[midway,worked load] {+0} (#1-2);
                \draw[worked edge] (#1-2) --
                    node[midway,worked load] {+0} (#1-3);
                \draw[worked edge] (#1-1) --
                    node[midway,worked load] {+0} (#1-3);
                \draw[worked edge] (#1-3) -- (#1-stub-left);
                \draw[worked edge] (#1-stub-right) -- (#1-4);
                \draw[worked edge] (#1-1) --
                    node[midway,worked load] {+0} (#1-4);
                \draw[worked edge] (#1-4) --
                    node[midway,worked load] {+0} (#1-5);
                \draw[worked edge] (#1-1) --
                    node[midway,worked load] {+0} (#1-5);
                \draw[worked edge] (#1-5) --
                    node[midway,worked load] {+0} (#1-6);
                \draw[worked edge] (#1-1) --
                    node[midway,worked load] {-1} (#1-6);
                \draw[worked edge] (#1-1) --
                    node[midway,worked load] {-1} (#1-7);
                \draw[worked edge] (#1-6) --
                    node[midway,worked load] {-1} (#1-7);
                \draw[worked edge] (#1-7) --
                    node[midway,worked load] {-1} (#1-8);
                \draw[worked edge] (#1-1) --
                    node[midway,worked load] {#5} (#1-8);
                \draw[worked edge] (#1-8) --
                    node[midway,worked load] {-2} (#1-9);
                \draw[worked edge] (#1-1) --
                    node[midway,worked load] {-2} (#1-9);
                \draw[worked edge] (#1-9) --
                    node[midway,worked load] {-2} (#1-10);
                \draw[worked edge] (#1-1) --
                    node[midway,worked load] {-2} (#1-10);
                \draw[worked edge] (#1-1) to[bend left=20]
                    node[midway,worked load] {-3} (#1-10);

                \foreach \i in {1,...,10}
                    \node[worked vertex] at (#1-\i) {};
                \node[fill=white,inner sep=0.8pt,font=\small] at (#1-dots)
                    {$\cdots$};
            \end{scope}%
        }

        \workedstate{old}{7.6}{old state}{32}{-1}

        % \draw[->,semithick] (4.25,7.2) -- (4.25,6.25);

        \begin{scope}[shift={(0,4.3)}]
            \node[anchor=east,font=\small] at (-0.55,0.5) {next pair};
            \coordinate (w1) at (4.25,2);
            \coordinate (w2) at (0,0);
            \coordinate (w3) at (1,0);
            \coordinate (w-stub-left) at (1.35,0);
            \coordinate (w-dots) at (1.75,0);
            \coordinate (w-stub-right) at (2.15,0);
            \coordinate (w4) at (2.5,0);
            \coordinate (w5) at (3.5,0);
            \coordinate (w6) at (4.5,0);
            \coordinate (w7) at (5.5,0);
            \coordinate (w8) at (6.5,0);
            \coordinate (w9) at (7.5,0);
            \coordinate (w10) at (8.5,0);

            \draw[worked first] (w1) to[bend right=20] (w2);
            \draw[worked first] (w2) -- (w3);
            \draw[worked first] (w3) -- (w-stub-left);
            \draw[worked first] (w-stub-right) -- (w4);
            \draw[worked first] (w4) -- (w5);
            \draw[worked first] (w1) -- (w6);
            \draw[worked first] (w6) -- (w7);
            \draw[worked first] (w8) -- (w9);
            \draw[worked first] (w1) -- (w9);
            \draw[worked first] (w1) to[bend left=20] (w10);

            \draw[worked second] (w1) -- (w2);
            \draw[worked second] (w1) -- (w3);
            \draw[worked second] (w1) -- (w4);
            \draw[worked second] (w1) -- (w5);
            \draw[worked second] (w5) -- (w6);
            \draw[worked second] (w1) -- (w7);
            \draw[worked second] (w7) -- (w8);
            \draw[worked second] (w9) -- (w10);
            \draw[worked second] (w1) -- (w10);

            \draw[worked unused] (w1) -- (w8);

            \foreach \i in {1,...,10}
                \node[worked vertex] at (w\i) {};
            \node[fill=white,inner sep=0.8pt,font=\small] at (w-dots)
                {$\cdots$};

            \node[font=\scriptsize,above=1.5pt] at (w1) {$1$};
            \node[font=\scriptsize,below=1.5pt] at (w2) {$2$};
            \node[font=\scriptsize,below=1.5pt] at (w3) {$3$};
            \node[font=\scriptsize,below=1.5pt] at (w4) {$n-6$};
            \node[font=\scriptsize,below=1.5pt] at (w5) {$n-5$};
            \node[font=\scriptsize,below=1.5pt] at (w6) {$n-4$};
            \node[font=\scriptsize,below=1.5pt] at (w7) {$n-3$};
            \node[font=\scriptsize,below=1.5pt] at (w8) {$n-2$};
            \node[font=\scriptsize,below=1.5pt] at (w9) {$n-1$};
            \node[font=\scriptsize,below=1.5pt] at (w10) {$n$};

            \draw[worked first] (1.2,-1.05) -- (1.7,-1.05)
                node[right,black,font=\scriptsize] {first tree};
            \draw[worked second] (3.7,-1.05) -- (4.2,-1.05)
                node[right,black,font=\scriptsize] {second tree};
            \draw[worked unused] (6.3,-1.05) -- (6.8,-1.05)
                node[right,black,font=\scriptsize] {unused edge};
        \end{scope}

        % \draw[->,semithick] (4.25,2.9) -- (4.25,2.25);

        \workedstate{new}{0}{new state}{34}{-2}
    \end{tikzpicture}
    \caption{Illustration of an induction step for $r=3$, Part II, with $a=1$: the old state, the next pair of trees, and the new state. The edge labels show $L^k_g-k/2$.}
    \label{fig:fan-worked-pair}
\end{figure}

After the last pair of phase $r$, the resulting state is $\langle z-4,4^{\times r}\rangle$. In phase $r+1$, the value of $z$ is smaller by four, so this is its first state. This completes the induction.
\end{proof}

\begin{proof}[Proof of \cref{th:lower-bound}]
We first bound the load of the rightmost edge $e'_n$ in $G_n$ using~\cref{lem:fan-phase-evolution}.

Consider any number $k$ of trees in phase $r$, so $4r(r-1)\leq k<4r(r+1)$. At the start of the phase, we have $L_{e'_n}^k - k/2 = -(r-1)$, and the first tree of the phase does not use $e'_n$. Before every later pair this difference is $-r$, and after the first tree of such a pair it is at most $-r+1/2$. Thus, throughout phase $r$,
\begin{equation}\label{eq:fan-last-edge-load}
    k\ell^k_{e'_n} = L_{e'_n}^k \leq\frac k2-(r-1).
\end{equation}

Suppose that $k<n^2/4$. Then $4r(r-1)\leq k<n^2/4$ implies $2r+1\leq n$, so \cref{lem:fan-phase-evolution} is applicable.
We have $r=\sqrt{k}/2+O(1)$, with an absolute hidden constant. Combining \eqref{eq:fan-ideal-load} and \eqref{eq:fan-last-edge-load}, for sufficiently large $k$ we obtain
\begin{equation}\label{eq:fan-last-edge-error-sharp}
    k\left(\ell^*_{e'_n}-\ell^k_{e'_n}\right)\geq r-1-\frac{k}{4n-2}
    =\frac{\sqrt{k}}{2}-\frac{k}{4n-2}+O(1)=\Omega(\sqrt{k}).
\end{equation}
Here, we used that $k<n^2/4$ gives $k/(4n-2)\leq\sqrt{k}/8+O(1)$, uniformly in $n$. We also need to argue the same lower bound for bounded $k$.\footnote{The analysis for bounded $k$ is, in our opinion, not conceptually important, but it is needed for the correctness of the proof, since we claim the bound for all $1 \leq k < n^2/4$.} It holds for bounded $k\geq8$, since $r-1>\sqrt{k}/5$ and $k/(4n-2)<\sqrt{k}/6$. For $1\leq k<8$, direct inspection of phase $1$ gives a load difference of at least $1/2-k/(4n-2)\geq1/6$.

It remains to extend the construction to larger values of the edge connectivity. Replace every edge of $G_n$ with $w$ copies. The resulting multigraph has edge connectivity $\lambda=3w$, and the ideal load of each copy of an edge $g$ is $\frac{1}{w}\left(\frac{1}{2}-\frac{1}{4n-2}\right)$. Replace every tree in the packing above by $w$ consecutive trees, where the $a$-th tree uses the $a$-th copy of every edge. This is still a greedy packing: during each block, an unused copy of every edge of the original minimum-weight tree retains its load from the start of the block, while no edge load decreases.

After packing $1\leq k<n^2/4$ trees, suppose first that $k\geq w$. Choose a copy of the rightmost edge that has not been used in the incomplete block, or any copy if the block is complete. Its load is the original edge's load after $\lfloor k/w\rfloor$ trees, whereas its ideal cumulative load is at least $\lfloor k/w\rfloor\left(\frac{1}{2}-\frac{1}{4n-2}\right)$. Applying the bound~\eqref{eq:fan-last-edge-error-sharp} at $\lfloor k/w\rfloor$ and using $\lfloor k/w\rfloor=\Theta(k/w)$ gives
\begin{equation}
    \|\ell^k-\ell^*\|_\infty=\Omega\left(\frac{\sqrt{k/w}}{k}\right)=\Omega\left(\frac{1}{\sqrt{k\lambda}}\right).
\end{equation}
Finally, if $k<w$, two copies of an edge in the first tree have loads differing by one, so one has relative error at least $1/(2k)$. Since $k<w$ and $\lambda=3w$, this is $\Omega(1/\sqrt{k\lambda})$, completing the proof.
\end{proof}

\paragraph*{Random tiebreaking.} Both the construction above and the construction from \cite{christiansen} rely on adversarial choices for the next tree at each iteration, since the MST is in general not unique. Therefore, there is hope that exploiting freedom in the choice of the next tree, for example, using random tiebreaking for the edges of the same load, could lead to better tree packings. We provide experimental evidence that this might not be the case.

We observe that on graphs $G_n$, the number of trees we need to pack to satisfy $\| \ell^k - \ell^* \|_\infty \leq \varepsilon$, is still of order $\sim 1/\varepsilon^2$, even if we randomize the choice of the next tree. 
The randomization we used is random tiebreaking in the edge weights. It is equivalent to setting edge weights to be the edge loads perturbed with independent infinitesimal noise at every iteration. 
We ran numerical experiments with $G_{100}, G_{300}, G_{1000}, G_{3000}$, packing $100, 900, 10^4, 9 \cdot 10^4$ trees respectively, which corresponds to $n^2 / 100$ trees. Empirically, the difference $\|k\ell^k-k\ell^*\|_\infty$ grows proportionally to $\sqrt{k}$ over the tested ranges.
Therefore, the error $\|\ell^k-\ell^*\|_\infty$ is proportional to $1/\sqrt{k}$ on the observed range. This behavior is consistent across four random seeds. Thus, our experiments suggest that random tiebreaking does not improve the convergence rate. We do not give a proof for the observed behavior; the evidence is purely experimental. 

We point out that random tiebreaking among edges of equal load is only one randomization strategy, so our experiments do not rule out the possibility that another randomization strategy could improve the $1/\varepsilon^2$ dependence.

For reproducibility, we share the code via the following link:
\href{https://github.com/Paul566/randomized-tree-packings-counterexample/blob/main/tree-packings.ipynb}{\nolinkurl{https://github.com/Paul566/randomized-tree-packings-counterexample/blob/main/tree-packings.ipynb}}.

\subsection{\texorpdfstring{A packing of $\Theta(\lambda^5 \log m)$ trees contains a tree crossing some min-cut once}{A packing of order lambda to the fifth times log m trees contains a tree crossing some min-cut once}}

Thorup~\cite[Theorem 9]{thorup} showed that $O(\lambda^7\log^3m)$ greedy trees suffice to contain a tree crossing some min-cut only once, a property used in his dynamic min-cut algorithm. We give a stronger bound:
\begin{theorem}
    \label{th_lambda_5_respect}
    A greedy tree packing with $k = \Theta(\lambda^5 \log m)$ trees contains a tree crossing some min-cut only once.
\end{theorem}
The proof follows Thorup's argument: assume that every packed tree crosses every min-cut at least twice, and consider the components formed by edges whose ideal loads lie below a threshold. Lowering this threshold refines the partition. We track the sum, over its parts, of the edge-boundary size minus $\lambda$. This sum does not decrease, and the assumption on min-cut crossings forces the number of parts to grow at a rate controlled by it. An upper bound on how many parts can be created within each initial part gives a contradiction after $O(\lambda)$ threshold steps when the load error is at most a sufficiently small constant times $1/\lambda^3$. Combining this with the known load-convergence bound gives $O(\lambda^5\log m)$ trees. The full proof is in~\cref{sec:th_lambda_5_respect_proof}.

\begin{remark}\label{rem:respect-lower-bound}
    There exist examples of greedy tree packings of $\Omega(\log n/\log\lambda)$ trees such that no tree crosses any min-cut only once.
\end{remark}
\begin{proof}
    For integers $\lambda\geq2$ and $h\geq2$, take a complete $\lambda$-ary tree of height $h$ rooted at a vertex $r$ and add a clique on its leaves. This gives a graph with $n=\Theta(\lambda^h)$ vertices. The cut isolating $r$ has size $\lambda$ and is the unique minimum cut.

    Break ties between edges of equal load by considering the original tree edges in increasing order of depth, followed by the clique edges. The first greedy spanning tree is the original rooted tree. We claim that after $j$ trees, every original tree edge whose parent has depth at most $h-j$ has load $j$, for $1\leq j\leq h$. This holds for $j=1$. If it holds for $j<h$, consider an edge whose parent has depth at most $h-j-1$ when Kruskal's algorithm constructs the next tree. Its child is an internal vertex, and all edges from this child to its children have load $j$ and come later in the order. Hence the child is still isolated when the edge is considered, and the edge is selected. This completes the induction step. In particular, each of the first $h$ trees contains all $\lambda$ edges incident to the root, so none 1-respects the unique minimum cut. Thus $\Omega(\log n/\log\lambda)$ trees may be needed.
\end{proof}

\paragraph*{AI disclosure.}
We used GPT-5.6 Sol and GPT-6 Astra to assist with the research and the paper draft. 

Originally, we used another version of the filtering algorithm that permanently removed elements with large load, and proved a convergence theorem for that. We then realized that it does not have a good recourse, after AI flagged that our proof for the recourse bound was incomplete. We then changed the algorithm, and AI helped to adapt the convergence proof by suggesting the recurrence in eq.~\eqref{eq:ABstep:a} and its solution. The proof of the sampling argument was drafted by AI.

We originally had another family of examples for the lower bound in \cref{th:lower-bound}. Through a series of interactions with AI trying to explore if randomizing the trees can make the bound more optimistic, we discovered an example that admitted a simpler proof of the lower bound for deterministic packings. The expressions for states in \Cref{tab:fan-deficit-states} were found by AI. The experimental evidence that the same construction shows the pessimistic error behavior even with a randomized process was also reported by AI. In the notebook accompanying the paper, the code for running simulations was written by human authors and then proofread by AI; the code making the illustrations was written by AI and edited by humans.

We used AI to draw all illustrations in the paper as tikz figures, as well as to improve the presentation, moving, extending and shortening parts of text. Finally, we used AI for proofreading the paper, which resulted in adding missed corner cases and correcting mathematical typos.

The authors reviewed and verified all mathematical claims, proofs, and code, and assume full responsibility for the content of the paper.

\bibliography{bibliography}

\appendix
\newpage

\section{Convergence of greedy base packing with filtering}\label{sec:th:filtered-greedy}
This section proves the following result.
\filteredGreedyConvergence*
It follows from the definition that $\rho=\rho(M)\ge 1$; we will implicitly use this fact below.

Fix a nonempty set $H\subseteq E$ attaining the density, so $|H|=\rho r(H)$. 
The hierarchy of ideal loads implies $\ell^S_e=1/\rho$ for every $e\in H$.
% For any $S$ with $H\subseteq S\subseteq E$, let $\ell^S$ denote the ideal loads of $M|S$. Since $\rho(M|S)=\rho$, the identity $\min_e \ell^S_e = 1 / \rho(M|S)$ implies $\ell^S_e\geq1/\rho$ for every $e\in S$. The base-polytope bound $\ell^S(H)\leq r(H)=|H|/\rho$ then forces $\ell^S_e=1/\rho$ for every $e\in H$.
Fix a value $\alpha\in\mathbb R$, and for an integer $k\ge 0$ and subset $S\subseteq E$ define
\begin{equation}
    \Phi_{k}(S) = \sum\limits_{e\in S} (1+\alpha)^{k\left(\frac{1}{\rho} - \ell^k_e \right)}.
\end{equation}
\begin{lemma} \label{lemma:ABstep}
If $H\subseteq  E_k$ for $k\ge 0$ then
\begin{subequations} \label{eq:ABstep}
\begin{align}
    \Phi_{k+1}(E) &\le \eta \cdot \Phi_{k}(E) + [k\ge \tau]\cdot m_{\max}\cdot\frac{\alpha(1+\alpha)^{1/\rho-1}}{\rho}\cdot(1+\alpha)^{-k/\rho} \quad && \forall \alpha > 0 \label{eq:ABstep:a}  \\
    \Phi_{k+1}(H) &\le \eta \cdot \Phi_k(H)  \quad                                && \forall \alpha \in (-1, 0) \label{eq:ABstep:b}
\end{align}
\end{subequations}
where
\begin{equation}
\eta=(1 + \alpha)^{1/\rho} \left( 1 - \frac{\alpha}{\rho(1 + \alpha)} \right)
\qquad\text{and}\qquad
\tau=6\rho^+\log m_{\max}.
\end{equation}
\end{lemma}
\begin{proof}
Let $v$ be the indicator vector of the next base $B_{k+1}$ in the packing. Then, $\ell^{k+1} = \frac{k\ell^k + v}{k+1}$. 
For a set $S\subseteq E$ let $\ell^S$ be the optimal loads vector in $M_S\eqdef M|S$.
Suppose that $H\subseteq S$, then $\rho(M_S)=\rho$ and hence $\ell^S_e \ge 1/\rho$ for any $e\in S$.
We will use the following derivation:
%===
%Let $v$ be the indicator vector of the next base $B_{k+1}$ in the packing. Then, $\ell^{k+1} = \frac{k\ell^k + v}{k+1}$. 
%Fix $S$ with $H\subseteq S\subseteq E$.
%We will use the following derivation:
    \begin{equation}
    \begin{split}
        \Phi_{k+1}(S) &= \sum\limits_{e\in S} (1+\alpha)^{(k+1)\left(\frac{1}{\rho} - \frac{k \ell_e^k + v_e^{}}{k+1} \right)} \\
        &= (1 + \alpha)^{(k + 1)/\rho} \sum\limits_{e\in S} (1+\alpha)^{-k \ell_e^k} (1 + \alpha)^{ - v_e^{}} \\
        &\stackrel{\text{(a)}}{=} (1 + \alpha)^{(k + 1)/\rho} \left( \sum\limits_{e\in S} (1+\alpha)^{-k \ell_e^k} - \sum\limits_{e\in S} v_e^{} \frac{\alpha}{1 + \alpha} (1+\alpha)^{-k \ell_e^k} \right) \\
        &\stackrel{\text{(b)}}{\le} (1 + \alpha)^{(k + 1)/\rho} \left( \sum\limits_{e\in S} (1+\alpha)^{-k \ell_e^k} - \sum\limits_{e\in S} \ell^S_e \frac{\alpha}{1 + \alpha} (1+\alpha)^{-k \ell_e^k} \right) \\
        &\stackrel{\text{(c)}}{\le} (1 + \alpha)^{(k + 1)/\rho} \left( \sum\limits_{e\in S} (1+\alpha)^{-k \ell_e^k} - \sum\limits_{e\in S} \frac 1 \rho \frac{\alpha}{1 + \alpha} (1+\alpha)^{-k \ell_e^k} \right) \\
        &= (1 + \alpha)^{1/\rho} \left( 1 - \frac{\alpha}{\rho(1 + \alpha)} \right)  \sum\limits_{e\in S} (1+\alpha)^{k \left(\frac 1 \rho - \ell_e^k\right)} \\
        &= \eta \Phi_k(S).
    \end{split}
    \end{equation}
    Here in (a) we used the fact that $v^{}$ is a $0/1$ vector, and hence $(1+\alpha)^{-v_e^{}} = 1 - v_e^{} \frac{\alpha}{1 + \alpha}$.
    Let us define  $w_e=-\frac{\alpha}{1 + \alpha} (1+\alpha)^{-k \ell_e^k}$,
    then the inequality in (b) is equivalent to the inequality 
    \begin{equation}\label{eq:b-inequality}
    \sum_{e \in S} w_e v_e \le \sum_{e \in S} w_e  \ell^S_e
    \end{equation}
    Note that if $\alpha>-1$ then vectors
     $\ell^k$ and $w$ induce the same ordering on the edges (since $w_e$ is a monotonic function of $\ell^k_e$), and hence $B_{k+1}$
is also a minimum-weight base w.r.t.\ weights $w$. 
    We claim that inequalities (b,c) hold in the following cases.
\begin{itemize}
\item $S=E_k$ and $\alpha>0$. Then~\eqref{eq:b-inequality} holds by the optimality of vector $v$,
and (c) holds since $\ell^S_e \ge 1/\rho$ for $e\in S$ and $\alpha>0$.
%Indeed, $B_{k+1}\subseteq E_k$ is a minimum-weight base w.r.t.\ weights $\ell^k$.
%Vectors $\ell^k$ and $w$ induce the same ordering on the edges, and hence $B_{k+1}$
%is also a minimum-weight base w.r.t.\ weights $w$. This implies the claim.
\item $S=H$ and $\alpha\in(-1,0)$.
Inequality (c) holds since $\ell^H_e=\frac 1 \rho$ for all $e\in H$. Let us show~\eqref{eq:b-inequality}.
Since $\ell^{E_k}(H)=r(H)$ and each base of $M_k\eqdef M|E_k$ meets $H$ in at most $r(H)$ elements, every base $B$ with positive coefficient in a convex combination representing $\ell^{E_k}$ has $B\cap H$ as a base of $M_k|H$. Moreover, $B_{k+1}\cap H$ can be extended to a minimum-weight base $B'$ of $M_k|H$:
run the matroid greedy algorithm with a tiebreaking order that produces $B_{k+1}$ and restrict the same order to $H$. Thus,
   \begin{equation}
       \sum_{e\in H}w_ev_e=\sum_{e\in B_{k+1}\cap H}w_e\leq\sum_{e\in B'}w_e\leq\sum_{e\in B\cap H}w_e,
   \end{equation}
   where the first inequality uses $w_e\geq0$ for $\alpha<0$. Averaging over the bases $B$ proves~\eqref{eq:b-inequality}.
\end{itemize}
We have proved \eqref{eq:ABstep:a} in the case when $k<\tau$ (since we then have $E_k=E$) and \eqref{eq:ABstep:b}.
Suppose that $k\ge \tau$ and $\alpha>0$. Denoting $\overline E_k=E-E_k$, we can write
%For each edge $e\in E-E_k$ we have $\ell^k_e > 2/\rho^-\ge 2/\rho$ and hence
%$ \ell^{k+1}_e\ge  \frac{k}{k+1}\ell_k \ge \frac{2k}{k+1}\cdot \frac 1 \rho$.
%Eq.~\eqref{eq:ABstep:a} now follows from the following bound:
\begin{equation}
\Phi_{k+1}(E)=\Phi_{k+1}(E_k)+\Phi_{k+1}(\overline E_k)\le\eta\Phi_{k}(E_k)+\Phi_{k+1}(\overline E_k)=\eta\Phi_k(E)+\Phi_{k+1}(\overline E_k)-\eta\Phi_{k}(\overline E_k)
\label{eq:GNLASJKHFAJKSGH}
\end{equation}
Now consider edge $e\in\overline E_k$. By construction, we have $\ell^k_e > 2/\rho^-\ge 2/\rho$.
We also have $(k+1)\ell^{k+1}_e\ge k\ell^k_e$, thus $\Phi_{k+1}(e)\le (1+\alpha)^{(k+1) /\rho - k\ell^{k}_e}=(1+\alpha)^{1/\rho}(1+\alpha)^{k \left(\frac 1 \rho - \ell_e^k\right)}$.
This yields
\begin{equation}
\begin{split}
\Phi_{k+1}(e)-\eta\Phi_{k}(e) %=(1+\alpha)^{(k+1) /\rho - k\ell^{k}_e}  -\eta  (1+\alpha)^{k \left(\frac 1 \rho - \ell_e^k\right)}
&\le\left( 
(1+\alpha)^{1 /\rho} 
- \eta\right) (1+\alpha)^{k \left(\frac 1 \rho -\ell_e^k\right)}
\\
&=
(1+\alpha)^{1 /\rho} 
\cdot \frac{\alpha}{\rho(1 + \alpha)}
\cdot (1+\alpha)^{k \left(\frac 1 \rho -\ell_e^k\right)}
\\
&\le
(1+\alpha)^{1 /\rho} 
\cdot \frac{\alpha}{\rho(1 + \alpha)}
\cdot (1+\alpha)^{k \left(\frac 1 \rho -\frac 2 \rho\right)}
\\
&=\frac{\alpha(1+\alpha)^{1/\rho-1}}{\rho}\cdot(1+\alpha)^{-k/\rho}
\end{split}
\end{equation}
Plugging this bound into~\eqref{eq:GNLASJKHFAJKSGH} establishes~\eqref{eq:ABstep:a}.
\end{proof}

\begin{lemma}
We have $H\subseteq E_k$ for every $k\ge 0$. Furthermore, if $k\ge \frac{6 \rho \log m_{\max}}{\varepsilon^2}$ then
\begin{equation}\label{eq:GLAJSFAG}
\ell^k_e\le \frac{1+\varepsilon}{\rho}\qquad\qquad\mbox{if }\varepsilon\in(0,1] \mbox{ and }e\in H
\end{equation}
%if $k\ge \frac{6 \rho \log m_{\max}}{\varepsilon^2}$ for $\varepsilon\in(0,1]$
%then $\ell^k_e\le \frac{1+\varepsilon}{\rho}$ for every $e\in H$.
\end{lemma}
\begin{proof}
Suppose that $H\subseteq E_1\cap\ldots\cap E_{k-1}$ for some $k\ge 1$.
We will prove that~\eqref{eq:GLAJSFAG} holds for such~$k$, assuming that $k\ge \frac{6 \rho \log m_{\max}}{\varepsilon^2}$.
This will imply, in particular, that $H \subseteq E_k$. (If this is false then $k\ge 6\rho^+\log m_{\max}$ and $\ell^k_e>2/\rho^-\ge 2/\rho$ for some $e\in H$,
which contradicts~\eqref{eq:GLAJSFAG} with $\varepsilon=1$). An induction argument will then establish the lemma.

Let us set $\alpha=-\beta$ for some $\beta\in(0,1)$. Using Lemma~\ref{lemma:ABstep} recursively yields
\begin{equation}
   \Phi_{k}(H) \le \eta^k\underbrace{\Phi_0(H)}_{\le m_{\max}} \le  (1-\beta)^{k/\rho}\left(1+\frac{\beta}{\rho(1-\beta)}\right)^k m_{\max} %\qquad\quad \forall \beta \in (0,1)
\end{equation}
Therefore, for each $e\in H$ we have
\begin{equation}
(1-\beta)^{k\left(\frac{1}{\rho} - \ell^k_e \right)} \le (1-\beta)^{k/\rho}\left(1+\frac{\beta}{\rho(1-\beta)}\right)^k m_{\max}
\end{equation}
Taking logarithms, dividing by $k$ and canceling terms gives
\begin{equation}
-\ell^k_e \cdot \log (1-\beta) \le  \log \left(1+\frac{\beta}{\rho(1-\beta)}\right) + \frac {\log m_{\max}} k
\end{equation}
Using inequality $\log(1 + t) \leq t$ for $t > -1$, we get
\begin{equation}
\ell^k_e \cdot \beta \le  \frac{\beta}{\rho(1-\beta)} + \frac {\log m_{\max}} k
\end{equation}
Suppose that \eqref{eq:GLAJSFAG} is false, i.e.\ $\ell_e^k > \frac{1+\varepsilon}{\rho}$.
We then have
\begin{equation}
\begin{split}
\frac{\log m_{\max}}{k} &\ge \frac{1+\varepsilon}{\rho}\cdot \beta - \frac{\beta}{\rho(1-\beta)} \\
&=\frac{\beta (\varepsilon-\beta-\varepsilon\beta)}{1-\beta}\cdot \frac 1 \rho.
\end{split}
\end{equation}
Setting $\beta=\varepsilon/4\in(0,1)$ gives $k<\frac{6 \rho\log m_{\max}}{\varepsilon^2}$, contradicting the choice of~$k$.

\end{proof}

\begin{lemma}
If $k\ge \frac{20 \rho\log (3m_{\max})}{\varepsilon^2}$ then
\begin{equation}\label{eq:IUIGAKSGH}
\frac 1{\ell^k_e}\le \rho(1+\varepsilon)\qquad\qquad\mbox{if }\varepsilon\in(0,1] \mbox{ and }e\in E
\end{equation}
\end{lemma}
\begin{proof}
Fix $\alpha\in(0,\frac 14]$, and denote 
\begin{equation}
\begin{aligned}
\eta&=(1 + \alpha)^{1/\rho} \left( 1 - \frac{\alpha}{\rho(1 + \alpha)} \right), \\
\sigma&=\frac{\alpha(1+\alpha)^{1/\rho-1}}{\rho}, \\
\gamma&=(1+\alpha)^{-1/\rho}.
\end{aligned}
\end{equation}
By~\cref{lemma:ABstep}, 
quantity $\Phi_k\eqdef\Phi_k(E)$ satisfies the following recurrence for each $k\ge 0$:
\begin{equation}
\Phi_{k+1}\le \eta \Phi_k + [k\ge \tau]\cdot m_{\max}\sigma \gamma^k
\end{equation}

We will need to show that $\eta>\gamma$, and the relative gap $\frac{\eta-\gamma}\sigma$ is sufficiently large.
Note that
$(1+\alpha)^{-x}\le 1-\frac{\alpha x}{(1+\alpha)}$ for $x\in(0,1)$, and hence
$
 (1+\alpha)^{-1/\rho}\le 1-\frac{\alpha}{\rho(1+\alpha)}
$.
Therefore,
\begin{equation}
\begin{split}
\frac{\alpha}{(1+\alpha)\sigma\rho}\cdot(\eta-\gamma)
&=1-\frac{\alpha}{\rho(1+\alpha)}-(1+\alpha)^{-2/\rho}
\\
&\ge
\left(1-\frac{\alpha}{\rho(1+\alpha)}\right)-\left(1-\frac{\alpha}{\rho(1+\alpha)}\right)^2
\\
&=
\left(1-\frac{\alpha}{\rho(1+\alpha)}\right)\cdot\frac{\alpha}{\rho(1+\alpha)}
\end{split}
\end{equation}
and so
\begin{equation}
\frac{\eta-\gamma}\sigma \ge 1-\frac{\alpha}{\rho(1+\alpha)}\ge \frac 34
\label{eq:eta-gamma}
\end{equation}

We claim that
\begin{equation}\label{eq:GALSG}
\Phi_k\le m_{\max} \eta^k + m_{\max}\sigma \sum_{j=\lceil\tau\rceil}^{k-1} \eta^{k-1-j}\gamma^j
\end{equation}
Indeed, this clearly holds for $k=\lceil\tau\rceil$ (since $\Phi_0=|E|\leq m_{\max}$). If $k\ge \lceil\tau\rceil$ and it holds for $k$ then it also holds for $k+1$ since
\begin{equation}
\begin{split}
\Phi_{k+1} &\le \eta \left( m_{\max} \eta^k + m_{\max}\sigma \sum_{j=\lceil\tau\rceil}^{k-1} \eta^{k-1-j}\gamma^j \right) + m_{\max}\sigma \gamma^k \\
&=m_{\max} \eta^{k+1} + m_{\max}\sigma \sum_{j=\lceil\tau\rceil}^{k} \eta^{k-j}\gamma^j.
\end{split}
\end{equation}
Denote $\lambda= \gamma /\eta$, then $\lambda<1$ by~\eqref{eq:eta-gamma}.
From \eqref{eq:GALSG} and~\eqref{eq:eta-gamma} we get
\begin{equation}
\begin{split}
\Phi_k &\le m_{\max} \eta^k \left(1 + \frac \sigma \eta \sum_{j=\lceil\tau\rceil}^{k-1} \lambda^j\right) \\
&< m_{\max} \eta^k \left(1 + \frac \sigma \eta \cdot\frac{1}{1-\lambda}\right) \\
&= m_{\max} \eta^k \left(1 + \frac {\sigma}{\eta -\gamma}\right) \\
&\le m_{\max}\eta^k \left(1 + \tfrac 43\right)<3m_{\max}\eta^k.
\end{split}
\end{equation}
For any edge $e\in E$ we have $(1+\alpha)^{k\left(\frac 1\rho-\ell^k_e\right)}\le \Phi_k$. Taking logarithms, dividing by $k$ and cancelling terms gives
\begin{equation}\label{eq:GLAKSHGAL}
-\ell^k_e \cdot \log (1+\alpha) \le  \log \left(1-\frac{\alpha}{\rho(1+\alpha)}\right) + \frac {\log (3m_{\max})} k
\end{equation}
Suppose that the lemma is false, i.e.\ $\ell_e^k < \frac{1}{\rho(1 + \varepsilon)}$. Plugging this into \eqref{eq:GLAKSHGAL} with $\alpha\in(0,\varepsilon)$ gives
\begin{equation}
\begin{split}
\frac{1}{\rho(1 + \varepsilon)} 
> \ell_e^k 
&\ge
-\frac {\log (3m_{\max})}{k\log (1+\alpha)}     - \frac {\log \left(1-\tfrac{\alpha}{\rho(1+\alpha)}\right)}{\log (1+\alpha)}
\\ &\ge
-\frac {\log (3m_{\max})}{k(\alpha-\alpha^2)}     + \frac {\alpha/(\rho(1+\alpha))}{\alpha} \\
&= \frac{1}{\rho(1+\alpha)} - \frac{\log (3m_{\max})}{k\alpha(1-\alpha)}.
\end{split}
\end{equation}
where we used inequalities $\log(1 + t) \leq t$ for $t > -1$ and $\log(1 + t) \geq t - \frac{t^2}{2}$ for $t \geq 0$.
We conclude that
$
k<\frac{(1+\alpha)(1+\varepsilon)}{\alpha(1-\alpha)(\varepsilon-\alpha)}\cdot\rho\log (3m_{\max})
$.
Setting $\alpha=\varepsilon/4$ and recalling that $\varepsilon\in(0,1]$ gives $k<\frac{20 \rho\log (3m_{\max})}{\varepsilon^2}$, contradicting the choice of~$k$.

\end{proof}

We are now ready to prove~\cref{th:filtered-greedy}.
Since $H\subseteq E_i$ for every $i$,
\begin{equation}
  \rho (M|E_1\cap\cdots\cap E_k)  \ge \rho(M|H)=\rho
\end{equation}
The reverse inequality follows because every subset of the intersection is
also a subset of $E$.  Hence equality holds.

Moreover, $B_j\cap H$ is independent in $M|H$, and therefore
\begin{equation}
  \sum_{e\in H}\ell_e^k
  =\frac1k\sum_{j=1}^k|B_j\cap H|
  \le r(H)=\frac{|H|}{\rho}
\end{equation}
Thus some $e\in H\subseteq E_k$ has $\ell_e^k\le1/\rho$, so
\begin{equation}
  \min_{e\in E_k}\ell_e^k\le\frac1\rho
\end{equation}
Together with the lower-load lemma,
\begin{equation}
  \frac1{\rho(1+\varepsilon)}
  \le\min_{e\in E_k}\ell_e^k
  \le\frac1\rho
\end{equation}
Taking reciprocals proves
\begin{equation}
  0\le
  \frac1{\min_{e\in E_k}\ell_e^k}-\rho
  \le\varepsilon\rho
\end{equation}

\section{Proof of the min-cut crossing bound}\label{sec:th_lambda_5_respect_proof}
We prove \cref{th_lambda_5_respect}.

We will closely follow some parts of Thorup's proof of his Theorem 9. Just like him, we will assume that every tree in a greedy MST packing crosses all min-cuts at least twice, and show that this leads to a contradiction, if $\|\ell^k - \ell^*\|_\infty \leq \frac{\varepsilon}{\lambda}$, for a suitably small value of $\varepsilon$. Thorup gives his bound for $\varepsilon = O\left(\frac{1}{\lambda^3 \log m}\right)$. We will show that having $\varepsilon = O\left(\frac{1}{\lambda^2}\right)$ is sufficient.

Following Thorup's line of proof, let $MC$ be the set of edges in min-cuts and let $\ell_0:=\min_{e\in MC}\ell_e^*$. Thorup's Lemma 18 shows that if the load error of a packing is at most $\varepsilon/\lambda$ and every tree in the packing crosses every min-cut at least twice, then
\begin{equation}
    \ell_0 > (1 - \varepsilon) \frac{2}{\lambda}.
    \label{eq_ell0_bound}
\end{equation}
For an integer $i\geq0$, define $\ell_i=\ell_0-i\varepsilon/\lambda$. For a real number $\ell$, let $E^*_{\geq\ell}:=\{e\in E:\ell^*_e\geq\ell\}$. Define $\mathcal{P}_i$ to be the partition whose sets are the components of $(V,E^*_{<\ell_i})$. Thorup proves that for any $i\geq2$, some trivial $\mathcal{P}_i$ cut is not a min-cut~\cite[Lemma 17]{thorup}.

Then, Thorup finds a parameter $a$ such that $|E^*_{\geq \ell_{a+2}} \setminus MC| \leq (1 + \alpha) |E^*_{\geq \ell_a} \setminus MC|$ for some small parameter $\alpha$. He upper-bounds $a$ by $O(\log_{1 + \alpha}m) = O\left( \frac{\log m}{\alpha} \right)$. This is where our proof diverges from Thorup's.

\begin{proof}[Proof of \Cref{th_lambda_5_respect}]

% Fix $\alpha = \varepsilon = \frac{1}{10 \lambda^2}$. We are going to show that $E^*_{\geq \ell_2}$ is somewhat close to $E^*_{\geq \ell_{\lambda}}$. This will allow us to make a more optimistic bound on $a$.

Let us assume that $\lambda \geq 2$. If $\lambda = 1$, then the first tree crosses all min-cuts once. Next, we need some definitions.
\begin{itemize}
    \item Let $G_i$ be the graph with the vertex set $V_i := \mathcal{P}_i$ and the edge set $E_i := E^*_{\geq \ell_i}$. Denote $n_i := |V_i|$, $m_i := |E_i|$.
    \item Let $p_i = |\{S \in \mathcal{P}_i \; | \; |\partial S| = \lambda\}|$ and $q_i = |\{S \in \mathcal{P}_i \; | \; |\partial S| > \lambda\}| = n_i - p_i$.
    \item Let $D_i = \sum\limits_{S \in \mathcal{P}_i} (|\partial S| - \lambda)$.
\end{itemize}
We will need the following fact.
\begin{lemma}[Lemma 14 in \cite{thorup}, reformulated and specialized]
    The values of $\ell^*_e$ are decreasing in the sense that for each $i$, for each $S \in \mathcal{P}_i$, we have $\ell^*_e < \ell_i$ for any edge $e \in E[S]$.
    \label{lem_x_decreasing}
\end{lemma}
Notice that $D_i = \sum\limits_{S \in \mathcal{P}_i \;:\; |\partial S| > \lambda} (|\partial S| - \lambda)$. Also, $D_i = 2 m_i - \lambda n_i$.
\begin{claim}
    The sequence $D_0, D_1, \ldots$ is non-decreasing.
    \label{claim_D_nondecreasing}
\end{claim}
\begin{proof}
    Notice that $D_i = (n_i - 1) \left( 2\frac{m_i}{n_i - 1} - \lambda \right) - \lambda$. Now we have that the sequence $D_i$ is non-decreasing if we show the following:
    \begin{enumerate}
        \item $n_i\geq2$
        \item $2\frac{m_i}{n_i - 1} - \lambda$ is positive
        \item $n_i$ is non-decreasing
        \item $2\frac{m_i}{n_i - 1} - \lambda$ is non-decreasing
    \end{enumerate}
    Items 1 and 3 follow because the partitions are nested and each $\mathcal P_i$ refines the sides of a min-cut. Item 2 follows from the fact that the minimum degree in $G_i$ is at least $\lambda$, so $\frac{2m_i}{n_i} \geq \lambda$.
    Item 4 follows from \cref{lem_x_decreasing}, since $\frac{m_i}{n_i - 1}$ is the reciprocal of the average $\ell^*_e$ for $e \in E_i$, and the average load does not increase.
\end{proof}

\begin{claim}
If $\varepsilon (i + 2)(\lambda + 1) \le 1$
then $|E[S] \cap E_i| \leq \lambda (|\partial S| - \lambda)$ for any $S \in \mathcal{P}_0$.
\end{claim}
\begin{proof}
We can assume that $E[S]\cap E_i$ is non-empty, since otherwise the statement is trivial.
Choose a value $\ell_S\le\ell_0$ (to be defined later), and let
 $H_S$ be a graph with the vertex set $S / E^*_{< \ell_S}$ and the edge set $E[S] \cap E^*_{\geq \ell_S}$.
 Let $H_S$ have $n_S \geq 2$ vertices and $m_S \geq 1$ edges. By the definition of $\ell_0$, strict subsets of $S$ cannot be min-cuts, and thus, they have the edge boundary of size at least $\lambda + 1$. Then, we can double-count $E(H_S)$ and have the following bound:
\begin{equation}
    (\lambda + 1) n_S \leq 2 m_S + |\partial S|.
    \label{eq_double_counting}
\end{equation}
Two cases are possible.
\begin{itemize}
\item $|\partial S| \le \lambda + 1$. We let $\ell_S := \max\limits_{e \in E[S]} \ell_e^*$.
    Using \eqref{eq_double_counting}, we can write
    \begin{equation}
        \ell_S = \frac{n_S - 1}{m_S} \leq \frac{2(n_S - 1)}{(\lambda + 1) n_S - |\partial S|} \leq \frac{2(n_S - 1)}{(\lambda + 1) n_S - (\lambda + 1)} = \frac{2}{\lambda + 1}.
    \end{equation}
The assumption
$\varepsilon (i + 2)(\lambda + 1) \le 1$ and \eqref{eq_ell0_bound}
give $\ell_i = \ell_0 - \frac{i \varepsilon}{\lambda} > \frac{2}{\lambda} - (i + 2) \frac{\varepsilon}{\lambda} > \frac{2}{\lambda} - \frac{2}{\lambda (\lambda + 1)} = \frac{2}{\lambda + 1}$.
We showed that the edges $e\in E[S]$ satisfy $\ell^*_e\le \ell_S \le \frac{2}{\lambda+1}$
while the edges $e\in E_i$ satisfy $\ell^*_e\ge \ell_i > \frac{2}{\lambda+1}$. Therefore, $E[S] \cap E_i$ is empty.

\item $|\partial S| > \lambda + 1$. We let $\ell_S := \ell_i$, then $E_S=E[S]\cap E_i$. Using \eqref{eq_ell0_bound} and \eqref{eq_double_counting}, we get
    \begin{equation}
(1 - \varepsilon) \frac{2}{\lambda} -  \frac{i \varepsilon} \lambda < \ell_0 - \frac{i \varepsilon} \lambda =\ell_i\leq \min\limits_{e \in E(H_S)} \ell^*_e \leq \avg_{e \in E(H_S)} \ell^*_e = \frac{n_S - 1}{m_S} \leq \frac{2 m_S + |\partial S| - \lambda - 1}{(\lambda + 1) m_S}.
    \end{equation}
    Rearranging and using the assumption $\varepsilon (i + 2)(\lambda + 1) \leq 1$ gives
    \begin{equation}
       \frac{|\partial S|-\lambda-1}{m_S} > \frac{\lambda+1}\lambda (2-(i+2)\varepsilon)-2 \ge \frac{\lambda+1}\lambda \left(2-\frac{1}{\lambda+1}\right)-2 = \frac 1\lambda
    \end{equation}
Therefore, $
        m_S \leq \lambda (|\partial S| - \lambda - 1) < \lambda (|\partial S| - \lambda)
    $,
    as claimed.

\end{itemize}
\end{proof}

\begin{corollary}
    \label{cor_ni_upper_bound}
    If $\varepsilon (i + 2)(\lambda + 1) \leq 1$ then
    \begin{equation}
        n_i \leq n_0 + 2 D_0.
    \end{equation}
\end{corollary}
\begin{proof}
    For $S\in\mathcal P_0$, let $n_S$ be the number of parts of $\mathcal P_i$ contained in $S$. If $|\partial S|\leq\lambda+1$, the first case in the preceding proof shows that $E[S]\cap E_i$ is empty, so $n_S=1$. If $|\partial S|>\lambda+1$, let $m_S=|E[S]\cap E_i|$; these are the vertex and edge counts of the graph $H_S$ obtained there with $\ell_S=\ell_i$. If $m_S=0$, then $n_S=1\leq2(|\partial S|-\lambda)$. If $m_S>0$, the preceding claim gives $m_S\leq\lambda(|\partial S|-\lambda)$. Moreover, every edge of $E_0$ has ideal load at least $\ell_0$, while its average ideal load is $(n_0-1)/m_0$. The minimum-degree bound for $G_0$ gives $(n_0-1)/m_0<2/\lambda$, and hence $\ell_0<2/\lambda$. By \cref{lem_x_decreasing}, the average ideal load in $H_S$ is therefore
    \begin{equation}
        \frac{n_S-1}{m_S}<\ell_0<\frac{2}{\lambda}.
    \end{equation}
    It follows, using integrality, that
    \begin{equation}
        n_S \leq 2 (|\partial S| - \lambda).
    \end{equation}
    Therefore,
    \begin{equation}
        n_i = p_0 + \sum\limits_{S \in \mathcal{P}_0 \;:\; |\partial S| > \lambda} n_S \leq n_0 + 2 D_0.
    \end{equation}
\end{proof}

Choose $k$ so that every greedy packing of $N\geq k$ trees satisfies $\|\ell^N-\ell^*\|_\infty\leq\varepsilon/(2\lambda)$, and set $k'=4k$. Suppose for contradiction that a greedy packing $\mathcal T'$ of $k'$ trees contains no tree that crosses a min-cut once, and let $\mathcal T$ be its prefix of $k$ trees. The relative load $\bar\ell$ of the suffix $\mathcal T'\setminus\mathcal T$ satisfies
\begin{equation}
    \begin{split}
        \|\bar\ell-\ell^*\|_\infty
        &=\left\|\frac{k'}{k'-k}(\ell^{k'}-\ell^*)-\frac{k}{k'-k}(\ell^k-\ell^*)\right\|_\infty\\
        &\leq\left(\frac{k'}{k'-k}+\frac{k}{k'-k}\right)\frac{\varepsilon}{2\lambda}
        =\frac{5\varepsilon}{6\lambda}<\frac{\varepsilon}{\lambda}.
    \end{split}
\end{equation}
Every tree in this suffix crosses every min-cut at least twice, so \eqref{eq_ell0_bound} applies.
\begin{claim}
    \begin{equation}
        n_{i + 1} - n_i > \frac{D_i}{\lambda} (1 - (i/2 + 3/2)(\lambda + 1) \varepsilon).
    \end{equation}
\end{claim}
\begin{proof}
    Let $T$ be a uniformly random tree from $\mathcal{T}'\setminus\mathcal{T}$. Every tree in the suffix is chosen at an iteration $N\geq k+1$, when its current load vector is $\ell^{N-1}$ and $\|\ell^{N-1}-\ell^*\|_\infty\leq\varepsilon/(2\lambda)$. If $g$ lies inside a part of $\mathcal P_{i+1}$ and $e\in E_i$, then $\ell_g^*<\ell_{i+1}=\ell_i-\varepsilon/\lambda$ and $\ell_e^*\geq\ell_i$. Hence $\ell_g^{N-1}<\ell_e^{N-1}$, so Kruskal's algorithm connects every part of $\mathcal P_{i+1}$ before using an edge of $E_i$. Consequently, $|T\cap E_i|\leq n_{i+1}-1<n_{i+1}$, and therefore
    \begin{equation}
        \mathbb{E}[|T \cap E_i|] < n_{i+1}.
        \label{eq_tcape_upper}
    \end{equation}

    On the other hand, double-counting $T \cap E_i$ leads us to
    \begin{equation}
    \begin{split}
        2\mathbb{E}[|T \cap E_i|] &= \mathbb{E}\left[ \sum\limits_{S \in \mathcal{P}_i} |T \cap \partial S| \right] =
        \sum\limits_{S \in \mathcal{P}_i \;:\; |\partial S| = \lambda} \mathbb{E}[|T \cap \partial S|] + \sum\limits_{S \in \mathcal{P}_i \;:\; |\partial S| > \lambda} \mathbb{E}[|T \cap \partial S|]  \\
        & \geq 2 p_i + \sum\limits_{S \in \mathcal{P}_i \;:\; |\partial S| > \lambda} (1 - (i/2 + 3/2) \varepsilon)\frac{2}{\lambda} |\partial S|.
        \label{eq_tcape_lower}
    \end{split}
    \end{equation}
    We used the fact that if $S$ is a min-cut, then $|T\cap\partial S|\geq2$. For every edge $e\in E_i$, the suffix-load bound and \eqref{eq_ell0_bound} give
    \begin{equation}
        \Pr[e\in T]=\bar\ell_e\geq\ell_e^*-\frac{\varepsilon}{\lambda}\geq\ell_0-\frac{(i+1)\varepsilon}{\lambda}>\left(1-(i/2+3/2)\varepsilon\right)\frac{2}{\lambda}.
    \end{equation}
    Putting together \eqref{eq_tcape_lower} and \eqref{eq_tcape_upper}, we get
    \begin{equation}
    \begin{split}
        n_{i+1} &> \mathbb{E}[|T \cap E_i|] \geq p_i + \sum\limits_{S \in \mathcal{P}_i \;:\; |\partial S| > \lambda} (1 - (i/2 + 3/2) \varepsilon)\frac{1}{\lambda} |\partial S|  \\
        &= p_i + q_i - q_i + \sum\limits_{S \in \mathcal{P}_i \;:\; |\partial S| > \lambda} \frac{|\partial S|}{\lambda} - \sum\limits_{S \in \mathcal{P}_i \;:\; |\partial S| > \lambda} (i/2 + 3/2) \varepsilon \frac{|\partial S|}{\lambda}  \\
        &= n_i + \frac{1}{\lambda} \sum\limits_{S \in \mathcal{P}_i \;:\; |\partial S| > \lambda} (|\partial S| - \lambda) - \sum\limits_{S \in \mathcal{P}_i \;:\; |\partial S| > \lambda} (i/2 + 3/2) \varepsilon \frac{|\partial S|}{\lambda}  \\
        &=n_i + \frac{D_i}{\lambda} - \sum\limits_{S \in \mathcal{P}_i \;:\; |\partial S| > \lambda} (i/2 + 3/2) \varepsilon \frac{|\partial S|}{\lambda},
    \end{split}
    \end{equation}
    which implies
    \begin{equation}
        n_{i + 1} - n_i > \frac{D_i}{\lambda} - \frac{(i/2 + 3/2) \varepsilon}{\lambda} \sum\limits_{S \in \mathcal{P}_i \;:\; |\partial S| > \lambda} |\partial S|.
    \end{equation}
    % Finally, plug in the bound from \cref{claim_bound_sum_ds}:
    % \begin{equation}
    %     n_{i + 1} - n_i > \frac{D_i}{\lambda} - 3 \lambda (i/2 + 1)^2 \varepsilon^2 n_i.
    % \end{equation}
    Finally, notice that $\sum\limits_{S \in \mathcal{P}_i \;:\; |\partial S| > \lambda} |\partial S| \leq (\lambda + 1) D_i$, and thus,
    \begin{equation}
        n_{i + 1} - n_i > \frac{D_i}{\lambda} (1 - (i/2 + 3/2)(\lambda + 1) \varepsilon).
    \end{equation}
\end{proof}

Finally, we are ready to reach the contradiction. Recall that the sequence $D_0,D_1,\ldots$ is non-decreasing by \cref{claim_D_nondecreasing}. Set $i=3\lambda$ and $\varepsilon=1/(10\lambda^2)$. Summing the bounds from the previous claim and using $D_j\geq D_0$, we have
\begin{equation}
    \begin{split}
        n_i
        &>n_0+\sum_{j=0}^{3\lambda-1}\frac{D_j}{\lambda}\left(1-(j/2+3/2)(\lambda+1)\varepsilon\right)\\
        &\geq n_0+D_0\left(3-\frac{3(\lambda+1)(3\lambda+5)}{40\lambda^2}\right).
    \end{split}
\end{equation}
The coefficient of $D_0$ is larger than $2$ for $\lambda\geq2$. Thus, if $D_0>0$, then $n_i>n_0+2D_0$. If $D_0=0$, Thorup's Lemma 17 gives $D_2>0$, so the same strict inequality follows from the $j=2$ summand. This contradicts \cref{cor_ni_upper_bound}; its hypothesis holds because $(3\lambda+2)(\lambda+1)/(10\lambda^2)\leq1$.

To have the desired value for $\varepsilon = \Theta\left(\frac{1}{\lambda^2}\right)$, we need $\Theta(\lambda \log m / \varepsilon^2) = \Theta(\lambda^5 \log m)$ trees. This concludes the proof.
\end{proof}

\section{Bicircular matroids: densest subgraph density}\label{sec:bicircular-matroids}

In this appendix, we consider the bicircular matroid of a graph $G = (V, E)$. The ground set of this matroid is $E$, and the independent sets are edge sets of all possible pseudoforests, i.e. such subgraphs that every connected component has at most one cycle. A base of the bicircular matroid is a maximal pseudoforest $P$. If $G$ has no acyclic components then $P$ can be characterized as follows: $P$ contains $n=|V|$ edges, and every component of $P$ contains exactly one cycle. The rank function of the bicircular matroid is $r(H) = |V(H)| - a(H)$, where $a(H)$ is the number of acyclic components of $H$, and $V(H)$ is the set of vertices adjacent to some edge in $H$.

Recall that the density of a graph is defined as $\rho(G)=\max_{\varnothing\ne S\subseteq V}\frac {|E(S)|}{|S|}$, and the subgraph $G[S]$ induced by a maximizer $S$ of this ratio is a \emph{densest subgraph of $G$}. This is related to the density of the bicircular matroid $M$ as follows.
\begin{observation}
    If $G$ is not a forest then $\rho(G)=\rho(M)$. Otherwise, if $G$ is a forest, then a largest component of $G$ is a densest subgraph.
\end{observation}
\begin{proof}
    If $G$ is a forest, every nonempty induced subgraph has density at most $1$, and a component with $N$ vertices has density $(N-1)/N$. Thus, a largest component is densest.

    Now suppose that $G$ contains a cycle. Let $S$ be the vertex set of a densest subgraph. Then $G[S]$ has no acyclic component, since deleting such a component would increase the density. Hence $r(E[S])=|S|$, and $\rho(M)\geq|E[S]|/r(E[S])=\rho(G)$. Conversely, consider any nonempty $H\subseteq E$ and decompose it into connected components. Each acyclic component contributes equally to $|H|$ and $r(H)$, whereas each cyclic component $C$ satisfies $|C|/r(C)=|C|/|V(C)|\leq\rho(G)$. Since $\rho(G)\geq1$, summing over the components gives $|H|/r(H)\leq\rho(G)$. Therefore, $\rho(M)=\rho(G)$.
\end{proof}
We obtain the following special case of the dynamic matroid-density algorithm. Here $m_{\max}\geq2$ is a fixed upper bound on the number of edges throughout the update sequence.

\begin{corollary}\label{cor:dynamic-graph-density}
Let $\varepsilon\in(0,1/2]$. There exists a deterministic data structure that maintains an estimate between $\rho(G)$ and $(1+\varepsilon)\rho(G)$ for a fully dynamic graph with worst-case update time
\begin{equation}
O\left((\rho_{\max}^2\varepsilon^{-2}+\rho_{\max}\varepsilon^{-4})\log^3m_{\max}\right),
\end{equation}
where $\rho_{\max}\geq1$ is a fixed upper bound on $\rho(G)$. Against an oblivious adversary, there is a randomized data structure that maintains a $(1\pm\varepsilon)$-approximation with worst-case update time $O(\varepsilon^{-6}\log^6m_{\max})$, whose estimates are correct with high probability throughout any polynomial-length update sequence.
\end{corollary}
\begin{proof}
Use the packing-maintenance algorithm of~\cref{sec:dynamic-matroid-density}, replacing the rank-oracle minimum-weight base primitive by the minimum-weight maximal pseudoforest data structure of~\cref{th:dynamic-pseudoforests}. A restriction of the bicircular matroid is the bicircular matroid of the retained subgraph. The same implementation takes $O(\log m_{\max})$ worst-case time per elementary update, including a weight change implemented by deletion and reinsertion, even without the polynomial-size assumption: it stores only $O(m_{\max})$ vertices and queue entries. The prefix-load and minimum-load bookkeeping has the same cost. Thus the base-update count in~\eqref{eq:dynamic-base-updates} gives the stated deterministic time bound, and the sampling reduction of~\cref{th_dynamic_density_sampling} gives the randomized bound.

It remains to handle forests, where graph density differs from bicircular-matroid density. Maintain one additional maximal pseudoforest for the original graph, augmented with the sizes of its acyclic components and the number of its cyclic components. These aggregates can be maintained with top trees and a heap in $O(\log m_{\max})$ worst-case time per update. The original graph is a forest if and only if the maintained pseudoforest has no cyclic component; in this case the pseudoforest is the whole graph. Return $(N-1)/N$, where $N$ is the number of vertices in its largest component, and return zero if the graph has no edges. Otherwise return the estimate of the bicircular-matroid density.
\end{proof}

These graph-density bounds do not improve the best known update times~\cite{SawlaniWang:STOC20,chekuri_density,ChekuriQuanrud:22}. They give an application of the general filtered-packing algorithm in which the rank-query bounds become actual update-time bounds.

\subsection{Dynamic minimum-weight pseudoforest}
\label{sec:dynamic-pseudoforest}

Let $n=|V|$, and let $m_{\max}\geq2$ be a fixed upper bound on the number of edges throughout the update sequence. We assume $m_{\max}=n^{O(1)}$, and write $\log n$ for $\log\max\{2,n\}$.
The following data structure may be useful independently. For comparison, minimum spanning forests can be maintained in $O(\log^4n)$ amortized time~\cite{Holm:01} or $n^{o(1)}$ Las Vegas worst-case time with high probability~\cite{NanongkaiSaranurakWulff-Nilsen2017}.

\begin{theorem}\label{th:dynamic-pseudoforests}
There exists a deterministic data structure for maintaining a minimum-weight maximal pseudoforest in a dynamically changing graph that uses $O(\log n)$ worst-case time per update (edge insertion or deletion).
\end{theorem}

\begin{proof}
We use the top-tree data structure~\cite{top-trees}. We do not store isolated vertices, so the number of stored vertices is at most $\min\{n,2m_{\max}\}$. The data structure supports the following operations.
\begin{itemize}
\item $\mathbf{tree}(v)$: returns a pointer to the tree containing $v$.
\item $\mathbf{link}(u,v, w)$: connects nodes $u,v$ of distinct trees with an edge $(u,v)$ of weight $w$.
\item $\mathbf{cut}(e)$: removes existing edge $e$.
\item $\mathbf{findmax}(u,v)$: returns an edge of maximum weight on the path between nodes $u,v$ in the same tree.
\item $\mathbf{findmin}(v)$: returns an edge of minimum weight in the tree containing node $v$.
\end{itemize}
All these operations take $O(\log n)$ worst-case time. The bounds for $\mathbf{findmax}$ and $\mathbf{findmin}$ follow from~\cite[Theorems 2.2 and 2.4]{top-trees}, respectively. The operation $\mathbf{tree}(v)$ follows from the implementation. It can also be simulated with $\mathbf{findmin}(v)$: after making edge weights distinct, the minimum-weight edge identifies each nontrivial tree; an isolated vertex identifies itself. Thus, we can test whether $u$ and $v$ belong to the same tree by comparing $\mathbf{tree}(u)$ and $\mathbf{tree}(v)$ in $O(\log n)$ time.

We now proceed with the proof of \cref{th:dynamic-pseudoforests}. Below we describe how updating graph $G=(V,E,w)$ affects a minimum-weight maximal pseudoforest $P$. For an edge set $F$, write $V(F)$ for its incident vertices; for a connected component, retain its vertex set even when it has no edges. We let $P'$ be the pseudoforest after the update and use a fixed ordering to break weight ties.
\begin{itemize}
\item \underline{Inserting edge $e=(u,v)$ with weight $w(e)$.}
If either $u$ or $v$ belongs to an acyclic component of $P$ then $P+e$ is still a pseudoforest, so $P'=P+e$. Otherwise for $w\in\{u,v\}$ let $C_w$ be the cycle of the component to which $w$ belongs united with the path from $w$ to this cycle. Define $F=C_u \cup C_v$. Clearly, $P+e-f$ is a pseudoforest for an edge $f\in P$ if and only if $f\in F$. Therefore, we can take $f\in\arg\max_{f\in F} w(f)$ and then set $P'=P+e-f$ if $w(e)< w(f)$, or $P'=P$ otherwise.
\item \underline{Removing edge $e$ from $G$.} If $e\notin P$ then $P'=P$.
Otherwise let $C$ be the component of $P$ containing $e$, and let $A$ be an acyclic component of $C-e$, retaining isolated vertices. If $e$ belongs to the cycle of $C$, then $A=C-e$. Otherwise, removing $e$ splits $C$ into two components, at least one of which is acyclic. If both are acyclic, either may be chosen: no edge of $E-P$ is incident to $V(C)$, by maximality of $P$. Let $F$ be the set of edges that are incident to a node in $V(A)$ and are not present in $P$. Clearly, $P-e+f$ is a pseudoforest for an edge $f\in E\setminus P$ if and only if $f\in F$. Therefore, if $F$ is empty then $P'=P-e$, otherwise take $f\in\arg\min_{f\in F} w(f)$ and then set $P'=P-e+f$.
\end{itemize}

To implement these operations efficiently, we will store the following information. For each component $C$ of $P$ we will maintain some edge $e_C$ on the cycle in $C$; if $C$ is acyclic then we let $e_C=\varnothing$. The tree $T_C=C-e_C$ will be stored in a top tree data structure $\calD$, with edge weights as in the original graph. The edge $e_C$ will be stored at the pointer corresponding to this tree. In addition, we will maintain a separate top tree data structure $\calD'$ that will store trees $T_C$ as above augmented with edges $(v,\bar v)$ for all nodes $v$, where $\bar v$ is a unique copy of $v$. At each edge $(v,\bar v)$ we will store a priority queue of edges incident to $v$ in $G$ and not present in $P$, with edge weights as the key. The minimum of these weights will be the weight of $(v,\bar v)$ in $\calD'$, while the weight of other edges in $\calD'$ will be $+\infty$.

We claim that after each graph modification described above data structures $\calD$ and $\calD'$ can be updated via a constant number of calls to operations on dynamic trees. Consider, for example, the case when edge $(u,v)$ is added to $G$ but not to $P$. We then remove edge $(u,\bar u)$ from $\calD'$, insert $(u,v)$ into the priority queue at $(u,\bar u)$, and then insert $(u,\bar u)$ back into $\calD'$ with the new weight. The same is done for $v$. All other operations are similar. Each priority queue has at most $m_{\max}$ edges, so an update takes $O(\log m_{\max})=O(\log n)$ time.

During edge insertion, a constant number of path-maximum queries in $\calD$, together with comparisons against the at most two stored cycle edges $e_C$, returns a maximum-weight edge of $F$. During edge deletion, a tree-minimum query in $\calD'$ returns a minimum-weight edge of $F$. Each graph update requires only a constant number of these queries and updates.
\end{proof}

\end{document}